\documentclass[cleveref,thm-restate]{lipics-v2021}

\pdfoutput=1
\hideLIPIcs
\nolinenumbers

\graphicspath{{./figures}}
\usepackage{wrapfig}

\usepackage{xparse}
\usepackage{xspace}

\usepackage{mathtools}

\newcommand{\etal}{\textit{et~al.}\xspace}

\newcommand{\f}{Fr\'echet\xspace}
\newcommand{\dF}{\ensuremath{d_\mathrm{F}}}
\newcommand{\F}{\ensuremath{\mathcal{F}}}

\newcommand{\bigO}{\ensuremath{\mathcal{O}}}
\newcommand{\eps}{\varepsilon\xspace}
\newcommand{\from}{\colon\xspace}
\newcommand{\R}{\ensuremath{\mathbb{R}}}

\newcommand{\calH}{\ensuremath{\mathcal{H}}}
\newcommand{\near}{\ensuremath{\mathrm{near}}}
\newcommand{\far}{\ensuremath{\mathrm{far}}}

\newcommand{\enter}{\ensuremath{\mathrm{enter}}}
\newcommand{\exit}{\ensuremath{\mathrm{exit}}}
\newcommand{\sep}{\ensuremath{\mathrm{sep}}}

\newcommand{\calC}{\ensuremath{\mathcal{C}}}
\newcommand{\calR}{\ensuremath{\mathcal{R}}}

\newcommand{\T}{\ensuremath{\mathcal{T}}}
\newcommand{\NN}{\ensuremath{\mathit{NN}}}
\newcommand{\V}{\ensuremath{\mathcal{V}}}

\newcommand{\rev}[1]{\reflectbox{$\vec{\reflectbox{\!$#1$}}$}}
\newcommand{\hor}{\ensuremath{\mathrm{hor}}}
\newcommand{\ver}{\ensuremath{\mathrm{ver}}}

\title{The Geodesic Fr\'echet Distance Between Two Curves Bounding a Simple Polygon}

\author
{Thijs van der Horst}
{Department of Information and Computing Sciences, Utrecht University, the Netherlands
\and
Department of Mathematics and Computer Science, TU Eindhoven, the Netherlands}
{t.w.j.vanderhorst@uu.nl}
{https://orcid.org/0009-0002-6987-4489}
{}

\author
{Marc van Kreveld}
{Department of Information and Computing Sciences, Utrecht University, The Netherlands}
{m.j.vankreveld@uu.nl}
{https://orcid.org/0000-0001-8208-3468}
{}

\author
{Tim Ophelders}
{Department of Information and Computing Sciences, Utrecht University, the Netherlands
\and
Department of Mathematics and Computer Science, TU Eindhoven, the Netherlands}
{t.a.e.ophelders@uu.nl}
{https://orcid.org/0000-0002-9570-024X}
{partially supported by the Dutch Research Council (NWO) under project no.\ VI.Veni.212.260.}

\author
{Bettina Speckmann}
{Department of Mathematics and Computer Science, TU Eindhoven, The Netherlands}
{b.speckmann@tue.nl}
{https://orcid.org/0000-0002-8514-7858}
{}

\authorrunning{T. van der Horst, M. van Kreveld, T. Ophelders, and B. Speckmann}
\Copyright{Thijs van der Horst, Marc van Kreveld, Tim Ophelders, and Bettina Speckmann}



\category{}

\relatedversion{}

\begin{document}

\maketitle

\begin{abstract}
    The Fr\'echet distance is a popular similarity measure that is well-understood for polygonal curves in~$\R^d$: near-quadratic time algorithms exist, and conditional lower bounds suggest that these results cannot be improved significantly, even in one dimension and when approximating with a factor less than three.
    We consider the special case where the curves bound a simple polygon and distances are measured via geodesics inside this simple polygon. Here the conditional lower bounds do not apply; Efrat~\textit{et~al.}~(2002) were able to give a near-linear time $2$-approximation algorithm.

    In this paper, we significantly improve upon their result: we present a $(1+\varepsilon)$-approximation algorithm, for any $\varepsilon > 0$, that runs in $\mathcal{O}(\frac{1}{\varepsilon} (n+m \log n) \log nm \log \frac{1}{\varepsilon})$ time for a simple polygon bounded by two curves with $n$ and $m$ vertices, respectively.
    To do so, we show how to compute the reachability of specific groups of points in the free space at once, by interpreting the free space as one between separated one-dimensional curves.
    We solve this one-dimensional problem in near-linear time, generalizing a result by
    Bringmann and K\"unnemann~(2015).
    Finally, we give a linear time exact algorithm if the two curves bound a convex polygon.
\end{abstract}

\section{Introduction}

    The \f distance is a well-studied similarity measure for curves in a metric space. Most results so far concern the \f distance between two polygonal curves $R$ and $B$ in $\R^d$ with $n$ and $m$ vertices, respectively. Then the \f distance between two such curves can be computed in $\tilde{\bigO}(nm)$ time (see e.g.~\cite{alt95continuous_frechet,buchin17continuous_frechet}). There is a closely matching conditional lower bound: If the \f distance between polygonal curves can be computed in $\bigO((nm)^{1-\eps})$ time (for any constant $\eps > 0$), then the Strong Exponential Time Hypothesis fails~\cite{bringmann14hardness}.
    This lower bound extends to curves in one dimension, and holds even when approximating to a factor less than three~\cite{buchin19seth_says}.

    Because it is unlikely that exact strongly subquadratic algorithms exist, approximation algorithms have been developed \cite{colombe21continuous_frechet,vanderhorst23continuous_frechet,cheng25constant_frechet}.
    Van~der~Horst~\etal~\cite{vanderhorst23continuous_frechet} were the first to present an algorithms which results in an arbitrarily small polynomial approximation factor ($n^\eps$ for any $\eps \in (0, 1]$) in strongly subquadratic time ($\tilde{\bigO}(n^{2-\eps})$).
    % This result was recently improved by van~der~Horst~and~Ophelders~\cite{vanderhorst24faster_approximation} who present an algorithm that runs in $\tilde{\bigO}(n^{2-3\eps})$ time in one dimension.
    Very recently, Cheng~\etal~\cite{cheng25constant_frechet} gave the first (randomized) constant factor approximation algorithm with a strongly subquadratic running time.
    Specifically, it computes a $(7+\eps)$-approximation in $\bigO(nm^{0.99} \log (n/\eps))$ time.

    For certain families of ``realistic'' curves, the SETH lower bound does not apply.
    For example, when the curves are \emph{$c$-packed}, Driemel~\etal~\cite{driemel12realistic} gave an $(1+\eps)$-approximation algorithm, for any $\eps \in (0, 1)$, that runs in $\tilde{\bigO}(cn/\eps)$ time.
    Bringmann~and~K\"unnemann~\cite{bringmann17cpacked} improved the running time to $\tilde{\bigO}(cn / \sqrt{\eps})$ time.

    For curves in one dimension with an imbalanced number of vertices, the \f distance can be computed in strongly subquadratic time without making extra assumptions about the shape of the curves.
    This was recently established by Blank~and~Driemel~\cite{blank24imbalanced_frechet}, who give an $\tilde{\bigO}(n^{2\alpha} + n)$-time algorithm when $m = n^\alpha$ for some $\alpha \in (0, 1)$.

    If the two polygonal curves $R$ and $B$ lie inside a simple polygon $P$ with $k$ vertices and we measure distances by the geodesic distance inside $P$, then neither the upper nor the conditional lower bound change in a fundamental way. Specifically, Cook~and~Wenk~\cite{cook10geodesic_frechet} show how to compute the \f distance in this setting in $\bigO(k + N^2 \log k N \log N)$ time, with $N = \max\{n, m\}$. For more general polygonal obstacles, Chambers~\etal~\cite{chambers10homotopic_frechet} give an algorithm that computes the \emph{homotopic} \f distance in $\bigO(N^9 \log N)$ time, where $N = m + n + k$ is the total number of vertices on the curves and obstacles.

    Har-Peled~\etal~\cite{harpeled16no_magic_leash} investigate the setting where $R$ and $B$ are simple, interior-disjoint curves on the boundary of a triangulated topological disk.
    If the disk has $k$ faces, their
    algorithm computes a $\bigO(\log k)$-approximation to the homotopic \f distance in $\bigO(k^6 \log k)$ time.
    Efrat~\etal~\cite{efrat02morphing} consider a more geometric setting,
    where $R$ and $B$ bound a simple
    \begin{wrapfigure}[6]{r}{0.18\textwidth}
        \raggedleft
        \vspace{-.6\baselineskip}\hspace{-2ex}
        \includegraphics{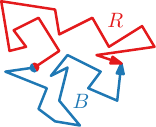}
    \end{wrapfigure}
    polygon (see figure).
    Here, the SETH lower bound does not apply; a $2$-approximation to the geodesic \f distance can be computed in $\bigO((n+m) \log nm)$ time~\cite{efrat02morphing}.
    Moreover, Van~der~Horst~\etal~\cite{vanderhorst25L_1_geodesic_frechet} recently gave an $\bigO((n+m) \log^4 nm)$-time exact algorithm for a similar setting, where distances are measured under the $L_1$-geodesic distance.
    Their result implies a $\sqrt{2}$-approximation algorithm for the geodesic \f distance.

\subparagraph*{Organization and results.} 
    In this paper, we significantly improve upon the results of Efrat~\etal~\cite{efrat02morphing} and Van~der~Horst~\etal~\cite{vanderhorst25L_1_geodesic_frechet}: we present a $(1+\varepsilon)$-approximation algorithm for the geodesic \f distance, for any $\varepsilon > 0$, that runs in $\bigO(\frac{1}{\eps} (n+m \log n) \log nm \log \frac{1}{\eps})$ time when $R$ and $B$ bound a simple polygon.
    We give an overview of our algorithm in~\cref{sec:outline}.
    Our algorithm relies on an interesting connection between matchings and nearest neighbors and is described in \cref{sec:approx}. There we also explain how to transform the decision problem for \emph{far} points on $B$ (those who are not a nearest neighbor of any point on $R$) into a problem between separated one-dimensional curves.
    Bringmann and K\"unnemann~\cite{bringmann17cpacked} previously solved the decision version of the Fr\'echet distance in this setting in $\bigO((n+m) \log nm)$ time. In~\cref{sec:separated_1D} we strengthen their result and compute the Fr\'echet distance between two separated one-dimensional curves in linear time.
    All omitted proofs can be found in the appendices.

    \begin{wrapfigure}[4]{r}{0.175\textwidth}
        \raggedleft
        \vspace{-.8\baselineskip}\hspace{-2ex}
        \includegraphics[page=1]{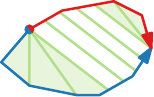}
    \end{wrapfigure}
    Finally, when $P$ is a convex polygon we describe a simple linear-time algorithm (relegated to \cref{app:convex}). We show that in this setting there is a \f matching with a specific structure: a \emph{maximally-parallel matching} (see figure). We compute the orientation of the parallel part from up to $O(n+m)$ different tangents, which we find using ``rotating calipers''.

\subparagraph{Preliminaries.}

    A (polygonal) \emph{curve} $R$ is a piecewise linear function that connects a sequence $r_1, \dots, r_n$ of points, which we refer to as \emph{vertices}.
    If the vertices lie in the plane, then we say $R$ is two-dimensional\footnote{%
        Curves are inherently one-dimensional objects. We abuse terminology slightly to refer to the ambient dimension as the dimension of a curve.
    } and equal to the function $R \from [1, n] \to \R^2$ where $R(i+t) = (1-t)r_i + tr_{i+1}$ for $i \in \{1, \dots, n-1\}$ and $t \in [0, 1]$.
    A one-dimensional curve is defined analogously.
    We assume $R$ is parameterized such that $R(i)$ indexes vertex $r_i$ for all integers $i \in [1, n]$.
    We denote by $R[x, x']$ the subcurve of $R$ over the domain $[x, x']$, and abuse notation slightly to let $R[r, r']$ to also denote this subcurve when $r = R(x)$ and $r' = R(x')$.
    The \emph{edges} of $R$ are the directed line segments $R[i, i+1]$ for integers $i \in [1, n-1]$.
    We write $|R|$ to denote the number of vertices of $R$.
    Let $R \from [1, n] \to \R^2$ and $B \from [1, m] \to \R^2$ be two simple, interior-disjoint curves with $R(1) = B(1)$ and $R(n) = B(m)$, 
    bounding a simple polygon~$P$.
    
    A \emph{reparameterization} of $[1, n]$ is a non-decreasing surjection $f \from [0, 1] \to [1, n]$.
    Two reparameterizations $f$ and $g$ of $[1, n]$ and $[1, m]$, describe a \emph{matching} $(f, g)$ between two curves $R$ and $B$ with $n$ and $m$ vertices, respectively, where any point $R(f(t))$ is matched to $B(g(t))$.
    The matching $(f, g)$ is said to have \emph{cost}
    \[
        \max_t~d( R(f(t)), B(g(t)) ),
    \]
    where $d(\cdot, \cdot)$ is the geodesic distance between points in $P$.
    A matching with cost at most $\delta$ is called a \emph{$\delta$-matching}.
    The (continuous) \emph{geodesic \f distance} $\dF(R, B)$ between $R$ and $B$ is the minimum cost over all matchings.
    The corresponding matching is a \emph{\f matching}.

    The \emph{parameter space} of $R$ and $B$ is the axis-aligned rectangle $[1, n] \times [1, m]$.
    Any point $(x, y)$ in the parameter space corresponds to the pair of points $R(x)$ and $B(y)$ on the two curves.
    A point $(x, y)$ in the parameter space is \emph{$\delta$-close} for some $\delta \geq 0$ if $d(R(x), B(y)) \leq \delta$.
    The \emph{$\delta$-free space} $\F_\delta(R, B)$ of $R$ and $B$ is the set of points $(x, y)$ in the parameter space with $d(R(x), B(y)) \leq \delta$.
    A point $q = (x', y') \in \F_\delta(R, B)$ is \emph{$\delta$-reachable} from a point $p = (x, y)$ if $x \leq x'$ and $y \leq y'$, and there exists a bimonotone (i.e., monotone in both coordinates) path in $\F_\delta(R, B)$ from $p$ to $q$.
    Points that are $\delta$-reachable from $(1, 1)$ are simply called $\delta$-reachable points.
    Alt and Godau~\cite{alt95continuous_frechet} observe that there is a one-to-one correspondence between $\delta$-matchings between $R[x, x']$ and $B[y, y']$, and bimonotone paths from $p$ to $q$ through $\F_\delta(R, B)$.
    We abuse terminology slightly and refer to such paths as $\delta$-matchings as well.
    
    Let $\eps > 0$ be a parameter.
    A \emph{$(1+\eps)$-approximate decision algorithm} for our problem takes a decision parameter $\delta \geq 0$, and outputs either that $\dF(R, B) \leq (1+\eps)\delta$ or that $\dF(R, B) > \delta$.
    It may report either answer if $\delta < \dF(R, B) \leq (1+\eps)\delta$.

\section{Algorithmic outline}
\label{sec:outline}

    In this section we sketch the major parts of our algorithm that approximates the geodesic \f distance $\delta_F:=\dF(R, B)$ between curves $R$ and $B$. 
    We approximate $\delta_F$ using a $(1+\eps)$-approximate decision algorithm, and use binary search to find the correct decision parameter. 
  In~\cref{sub:approx_optimization} we relate the \f distance to the geodesic Hausdorff distance $\delta_H$, which allows us to bound the number of iterations of the binary search, 
    In particular, we show that $\delta_F$ lies in the range $[\delta_H, 3\delta_H]$.
    This range can be computed in $\bigO((n+m) \log nm)$ time~\cite[Theorem~7.1]{cook10geodesic_frechet}.
    A binary search over this range results in a $(1+\eps)$-approximation of the \f distance after $\bigO(\log \frac{1}{\eps})$ calls to the decision algorithm.
    \Cref{thm:ourmaintheorem} follows.

    \begin{restatable}{theorem}{ourmaintheorem}\label{thm:ourmaintheorem}
        For any $\eps > 0$, we can compute a $(1+\eps)$-approximation to $\dF(R, B)$ in $\bigO(\frac{1}{\eps} (n+m \log n) \log nm \log \frac{1}{\eps})$ time.
    \end{restatable}

\subparagraph*{The approximate decision algorithm.}
    In the remainder of this section we outline the approximate decision algorithm, which is presented in detail in \cref{sec:approx}.    
    At its heart lies a useful connection between matchings and nearest neighbors.
    A \emph{nearest neighbor} of a point $r$ on $R$ is any point $b$ on $B$ that among all points on $B$ is closest to $r$.
    We denote the nearest neighbor(s) of $r$ by $\NN(r)$.
    We prove in~\cref{sub:fans_and_nearest_neighbors} that any $\delta$-matching matches each nearest neighbor $b$ of $r$ relatively close to $r$.
    Specifically, $b$ must be matched to a point $r'$ for which the entire subcurve of $R$ between $r$ and $r'$ is within distance $\delta$ of $b$.
    We introduce the concept of a $(r, b, \delta)$-nearest neighbor fan to capture the candidate locations for $r'$.
    
    The \emph{$(r, b, \delta)$-nearest neighbor fan} $F_{r, b}(\delta)$ consists of the point $b$ and the maximal subcurve $R[x, x']$ that contains $r$ and is within geodesic distance $\delta$ of $b$; it is the union of geodesics between $b$ and points on $R[x, x']$, see~\cref{fig:nearest_neighbor_fans}.
    We call $b$ the \emph{apex} of $F_{r, b}(\delta)$ and $R[x, x']$ the \emph{leaf} of the fan.
    We prove in~\cref{sub:fans_and_nearest_neighbors} that any $\delta$-matching must match the apex $b$ to a point in the leaf $R[x, x']$.
    
    As $r$ moves forward along $R$, so do its nearest neighbors $b$ along $B$.
    Their nearest neighbor fans $F_{r, b}(\delta)$ move monotonically through the polygon $P$ bounded by $R$ and $B$.
    However, while~$r$ moves continuously along $R$, the points $b$ and their nearest neighbor fans might jump discontinuously.
    Such discontinuities occur due to points $b$ that are not a nearest neighbor of any point on $R$, and thus at points that are not the apex of a nearest neighbor fan.
    We classify the points on $B$ accordingly:  
    we call a point $b$ on $B$ a \emph{near} point if it is a nearest neighbor of at least one point on $R$, and call $b$ a \emph{far} point otherwise.

    \begin{figure}
        \centering
        \includegraphics{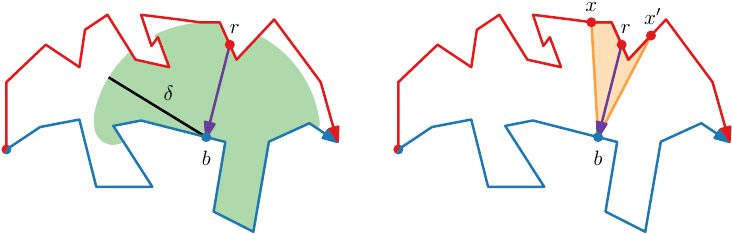}
        \caption{(left) Points $r$ and $b$ with $b \in \NN(r)$.
        The region shaded in green consists of all points within geodesic distance $\delta$ of $b$.
        (right) The $(r, b, \delta)$-nearest neighbor fan (orange).}
        \label{fig:nearest_neighbor_fans}
    \end{figure}
    
    \begin{figure}[b]
        \centering
        \includegraphics[page=1]{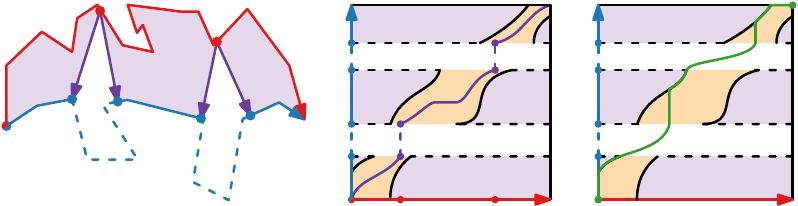}
        \caption{(left) The mappings (purple) from points on $R$ to their nearest neighbor(s) on $B$.
        (middle) The partition of the parameter space based on near and far points on $B$.
        The partly-dashed purple curve indicates the nearest neighbor(s) on $B$ of points on $R$.
        The beige regions correspond to the $(r, b, \delta)$-nearest neighbor fans.
        (right) A $\delta$-matching that is greedy on $B$ inside the regions.}
        \label{fig:near_and_far_free_space}
    \end{figure}

    The distinction between near and far points induces a partition $\calH$ of the parameter space into horizontal slabs.
    We consider these slabs from bottom to top, and iteratively construct a $\delta$-matching (provided that one exists).
    Recall that a $\delta$-matching is a bimonotone path in $\F_\delta(R, B)$ from the bottom-left corner of the parameter space to the top-right corner.

    Inside a slab $H_\near \in \calH$ corresponding to a subcurve of $B$ with only near points, the nearest neighbor fans correspond to a connected, $x$-monotone and $y$-monotone region $\calR$ spanning the entire height of $H_\near$.
    These regions are illustrated in~\cref{fig:near_and_far_free_space}.
    The intersection between any $\delta$-matching and $H_\near$ is contained in $\calR$.
    The structure of $\calR$ implies that if a $\delta$-matching between $R$ and $B$ exists, there exists one which moves vertically up inside $\calR$ whenever this is possible.
    Geometrically, this corresponds to greedily traversing the near points on $B$, and traversing parts of $R$ only when necessary.
    We formalize this in~\cref{sub:near_points}.

    Slabs whose corresponding subcurves of $B$ have only far points pose the greatest technical challenge for our algorithm; we show how to match far points in an approximate manner in~\cref{sub:far_points}.
    Specifically, let $H_\far \in \calH$ correspond to a subcurve $B[b, b']$ with only far points on its interior.
    We compute a suitable subcurve of $R$ that can be $(1+\eps)\delta$-matched to $B[b, b']$ in the following manner.
    First we argue that $d(b,b') \leq 2\delta_F$.
    In other words, the geodesic from $b$ to $b'$ is short and separates $R$ from the subcurve $B[b,b']$.
    We use this separating geodesic to transform the problem of creating a matching for far points into $K=\bigO(1/\eps)$ one-dimensional problems.

    \begin{figure}
        \centering
        \includegraphics[page=1]{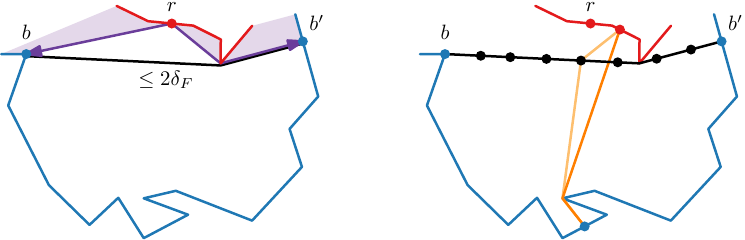}
        \caption{(left) The points $b$ and $b'$ are both nearest neighbor of some point $r$, implying a short separator.
        (right) Adding anchor points to the separator and snapping the orange geodesic (between arbitrary points on $R$ and $B[b, b']$) to one.}
        \label{fig:anchors}
    \end{figure}
    
    Specifically, we discretize the separator with $K$ points, which we call \emph{anchors}, and ensure that consecutive anchors are at most $\eps \delta$ apart.
    We snap our geodesics to these anchors (see~\cref{fig:anchors}), which incurs a small approximation error.
    Based on which anchor a geodesic snaps to, we partition the parameter space of $R$ and $B[b, b']$ into regions, one for each anchor.
    For each anchor point, the lengths of the geodesics snapped to it can be expressed as distances between points on two separated one-dimensional curves; this results in a one-dimensional problem which we can solve exactly.
    In \cref{sub:far_points} we present the transformation into one-dimensional curves, and in \cref{sec:separated_1D} we present an efficient exact algorithm for the resulting one-dimensional problem.

\section{Approximate geodesic Fr\'echet distance}
\label{sec:approx}

\subsection{Nearest neighbor fans and partitioning the parameter space}
\label{sub:fans_and_nearest_neighbors}

    We first present useful properties of nearest neighbor fans and their relation to matchings.
    In~\cref{lem:matching_to_fan} we prove a crucial property of nearest neighbor fans, namely that \emph{any} $\delta$-matching between $R$ and $B$ matches $b$ to a point on the leaf of the fan.
    For the proof, we make use of the following auxiliary lemma:

    \begin{lemma}
    \label{lem:shortcutting}
        Let $r \in R$ and $b \in \NN(r)$.
        For any points $r' \in R$ and $b' \in B$ on opposite sides of the geodesic $\pi(r, b)$ between $r$ and $b$, we have $d(r', b) \leq d(r', b')$.
    \end{lemma}
    \begin{proof}
            All points $p$ on $\pi(r, b)$ naturally have the property that $d(p, b) \leq d(p, b')$ for all $b' \in B$.
            For any $r' \in R$ and $b' \in B$ on opposite sides of $\pi(r, b)$, the geodesic $\pi(r', b')$ intersects $\pi(r, b)$ in a point $p$.
            It follows from the triangle inequality that
            \begin{align*}
                d(r', b) &\leq d(r', p) + d(p, b) \\
                &\leq d(r', p) + d(p, b') = d(r', b').\qedhere
            \end{align*}
    \end{proof}

    \begin{lemma}
    \label{lem:matching_to_fan}
        Let $r \in R$ and $b \in \NN(r)$.
        For any $\delta \geq 0$, every $\delta$-matching between $R$ and $B$ matches $b$ to a point in the leaf of $F_{r, b}(\delta)$.
    \end{lemma}
    \begin{proof}
        Suppose $b$ is matched to a point $r'$ by some $\delta$-matching $(f, g)$.
        Assume without loss of generality that $r'$ comes before $r$ along $R$.
        Let $\hat{r}$ be a point between $r'$ and $r$.
        The $\delta$-matching $(f, g)$ matches $\hat{r}$ to a point $\hat{b}$ after $b$.
        Thus we obtain from~\cref{lem:shortcutting} that $d(\hat{r}, b) \leq d(\hat{r}, \hat{b}) \leq \delta$.
        This proves that all points between $r'$ and $r$ are included in the leaf of $F_{r, b}(\delta)$.
    \end{proof}

    \noindent
    We partition the parameter space into (closed) maximal horizontal slabs, such that for each slab $[1, n] \times [y, y']$, either the subcurve $B[y, y']$ contains only near points, or its interior contains only far points.
    Let $\calH$ be the resulting partition.
    Each slab $H \in \calH$ has two horizontal line segments, one on its bottom and one on its top side, that correspond to nearest neighbor fans.
    We refer to these line segments as the \emph{entrance} and \emph{exit intervals} of $H$, and refer to points on them as the \emph{entrances} and \emph{exits} of $H$.
    A consequence of~\cref{lem:matching_to_fan} is that any $\delta$-matching enters and leaves $H$ through an entrance and exit.
    We compute \emph{$\delta$-safe} entrances and exits: $\delta$-reachable points from which $(n, m)$ is $\delta$-reachable.
    Each such point is used by a $\delta$-matching, and is used to iteratively determine if such a matching exists.
    
    It proves sufficient to consider only a discrete set of entrances and exits for each slab.
    For each slab, we identify (implicitly) a set of $\bigO(n)$ entrances and exits that contain $\delta$-safe entrances and exits (if any exist at all).
    We define these entrances and exits using locally closest points and will call them transit points.
    
    A point $r$ on $R$ is \emph{locally closest} to a point $b$ on $B$ if perturbing $r$ infinitesimally while staying on $R$ increases its distance to $b$.
    The \emph{transit} entrances and exits are those entrances and exits $(x, y)$ where $R(x)$ is either a vertex or locally closest to $B(y)$.
    We show that it is sufficient to consider only transit entrances and exits:

    \begin{lemma}
    \label{lem:critical_exits}
        If there exists a $\delta$-matching, then there exists one that enters and leaves each slab through a transit entrance and exit.
    \end{lemma}
    \begin{proof}
        Suppose $\dF(R, B) \leq \delta$.
        We prove that there exists a $\delta$-matching between $R$ and $B$ that matches each point $b$ on $B$ to either a vertex of $R$, or a point locally closest to $b$.
                
        Let $(f, g)$ be a $\delta$-matching between $R$ and $B$.
        Based on $(f, g)$, we construct a new $\delta$-matching $(f', g')$ that satisfies the claim.
        Moreover, the matching $(f', g')$ has the property that for every matched pair $(r, b)$, the point $r$ is either a vertex or locally closest to $b$.

        Let $r_1, \dots, r_n$ and $b_1, \dots, b_m$ be the sequences of vertices of $R$ and $B$, respectively.
        For each vertex $r_i$, if $(f, g)$ matches it to a point interior to an edge $\overline{b_j b_{j+1}}$ of $B$, or to the vertex $b_j$, then we let $(f', g')$ match $r_i$ to the point on $\overline{b_j b_{j+1}}$ closest to it.
        Symmetrically, for each vertex $b_j$, if $(f, g)$ matches it to a point interior to an edge $\overline{r_i r_{i+1}}$ of $R$, or to the vertex $r_{i+1}$, then we let $(f', g')$ match $b_j$ to the point on $\overline{r_i r_{i+1}}$ closest to it.
        This point is either a vertex of $R$ or locally closest to $b_j$.

        Consider two maximal subsegments $\overline{r r'}$ and $\overline{b b'}$ of $R$ and $B$ where currently, $r$ is matched to $b$ and $r'$ is matched to $b'$.
        See~\cref{fig:restricted_matching} for an illustration of the following construction.
        Let $\overline{r r'} \subseteq \overline{r_i r_{i+1}}$ and $\overline{b b'} \subseteq \overline{b_j b_{j+1}}$.
        We have $r = r_i$ or $b = b_j$, as well as $r' = r_{i+1}$ or $b' = b_{j+1}$.

        Let $\hat{r}, \hat{r}' \in \overline{r r'}$ be the points closest to $b$ and $b'$, respectively.
        We let $(f', g')$ match $b$ to $\overline{r \hat{r}}$ and $b'$ to $\overline{\hat{r}' r'}$.
        The maximum distance from $b$ to a point on $\overline{r \hat{r}}$ is attained by $r$~\cite{cook10geodesic_frechet}, and symmetrically, the maximum distance from $b'$ to a point on $\overline{\hat{r}' r'}$ is attained by $r'$.
        Thus, these matches have a cost of at most $\dF(R, B)$.

        To complete the matching, we match $\overline{\hat{r} \hat{r}'}$ to $\overline{b b'}$ such that each point on $\overline{b b'}$ gets matched to the point on $\overline{\hat{r} \hat{r}'}$ closest to it.
        This is a proper matching, as the closest point on $\overline{\hat{r} \hat{r}'}$ moves continuously along the segment as we move continuously along $\overline{b b'}$.
        Since the original matching $(f, g)$ matches the segment $\overline{b b'}$ to part of $\overline{r_i r_{i+1}}$, these altered matchings do not increase the cost.
        This matching thus has cost at most $\dF(R, B)$.
    \end{proof}
    
    \noindent
    Our algorithm computes a $\delta$-safe transit exit for each slab.
    To do so, it requires the explicit entrance and exit intervals.
    We compute these using a geodesic Voronoi diagram.

    \begin{lemma}
    \label{lem:constructing_partition}
        The partition $\calH$ consists of $\bigO(m)$ slabs.
        We can compute $\calH$, together with the entrance and exit interval of each slab, in $\bigO((n+m) \log nm)$ time in total.
    \end{lemma}
    \begin{proof}
        We first construct the \emph{geodesic Voronoi diagram} $\V_P(B)$ of $B$ inside $P$.
        This diagram is a partition of $P$ into regions containing those points for which the closest edge(s) of $B$ (under the geodesic distance) are the same.
        Points inside a cell have only one edge of $B$ closest to them, whereas points on the segments and arcs bounding the cells have multiple.
        The geodesic Voronoi diagram can be constructed in $\bigO((n+m) \log nm)$ time with the algorithm of Hershberger~and~Suri~\cite{hershberger99optimal_Euclidean_paths}.

        The points on $R$ that lie on the boundary of a cell of $\V_P(B)$ are precisely those that have multiple (two) points on $B$ closest to them.
        By general position assumptions, these points form a discrete subset of $R$.
        Furthermore, there are only $\bigO(m)$ such points.
        We identify these points $r$ by scanning over $\V_P(B)$.
        From this, we also get the two edges of $B$ containing $\NN(r)$, and we compute the set $\NN(r)$ in $\bigO(\log nm)$ time with the data structure of~\cite[Lemma~3.2]{cook10geodesic_frechet} (after $\bigO(n+m)$ preprocessing time).
        The sets $\NN(r)$ with cardinality two are precisely the ones determining the partition $\calH$.
        The partition consists of $\bigO(m)$ slabs.

        We first compute the right endpoints for the entrance intervals (and thus of the exit intervals) of slabs in $\calH$.
        The left endpoints can be computed by a symmetric procedure.
        
        We start by computing the last point in the leaf of $F_{R(1), B(1)}$.
        This point corresponds to the right endpoint of the entrance interval of the bottommost slab.
        To compute this endpoint, we determine the first vertex $r_i$ of $R$ with $d(r_i, B(1)) > \delta$.
        We do so in $\bigO(i \log nm)$ time by scanning over the vertices of $R$.
        For any edge $\overline{r_{i'} r_{i'+1}}$ of $R$, the distance $d(r, B(1))$ as $r$ varies from $r_{i'}$ to $r_{i'+1}$ first decreases monotonically to a global minimum and then increases monotonically~\cite[Lemma~2.1]{cook10geodesic_frechet}.
        Thus $d(r, B(1)) \leq \delta$ for all $r \in R[1, i-1]$.
        Moreover, the last point in the leaf of $F_{R(1), B(1)}(\delta)$ is the last point on $\overline{r_{i-1} r_i}$ with geodesic distance $\delta$ to $B(1)$.
        We compute this point in $\bigO(\log nm)$ time with the data structure of~\cite[Lemma~3.2]{cook10geodesic_frechet}.

        Let $H = [1, n] \times [1, y']$ be the bottommost slab of $\calH$.
        Next we compute the right endpoint of the entrance interval of the slab directly above $H$.
        Suppose $B(y') \in \NN(r)$.
        Let $r^* = R(x^*)$ be the last point in the leaf of $F_{R(1), B(1)}(\delta)$, so $(x^*, 1)$ is the right endpoint of the entrance interval of $H$.
        From the monotonicity of the fan leaves~(\cref{lem:monotone_leaves}), the last point in the leaf of $F_{r, B(y)}(\delta)$ comes after $r^*$ along $R$.
        Let $r = R(x)$.
        Given that $r$ is contained in the leaf of $F_{r, B(y)}(\delta)$ (by our assumption that there are no empty nearest neighbor fans), the endpoint we are looking for lies on $[\max\{x^*, x\}, n] \times \{y'\}$.
        We proceed as before, setting $R$ to be the subcurve $R[\max\{x^*, x\}, n]$ and $B$ to be the subcurve $B[y', m]$.

        The above iterative procedure takes $\bigO((n+|\calH|) \log nm)$ time, where $|\calH| = \bigO(m)$ is the number of slabs in $\calH$.
        This running time is asymptotically the same as the time taken to construct $\calH$.
    \end{proof}

\subsection{Slabs of near points}
\label{sub:near_points}

    Let $H_\near \in \calH$ be a slab corresponding to a subcurve $\hat{B}$ of $B$ with only near points.
    We use properties of nearest neighbor fans to determine a $\delta$-safe transit exit of $H_\near$.
    A crucial property is that the nearest neighbor fans behave monotonically with respect to their apexes, if $\delta$ is large enough.
    Specifically, this is the case if $\delta \geq \delta_H$, the geodesic Hausdorff distance between $R$ and $B$.
    This is the maximum distance between a point on $R \cup B$ and its closest point on the other curve.

    \begin{lemma}
    \label{lem:monotone_leaves}
        Suppose $\delta \geq \delta_H$.
        Let $b$ and $b'$ be near points on $B$ and let $R[x_1, x'_1]$ and $R[x_2, x'_2]$ be the leaves of their respective nearest neighbor fans.
        If $b$ comes before $b'$, then $x_1 \leq x_2$ and $x'_1 \leq x'_2$.
    \end{lemma}
    \begin{proof}
        The assumption that $\delta \geq \delta_H$ implies that the distance between any point on $R$ and their nearest neighbor is at most $\delta$, thus ensuring that the leaves $R[x_1, x'_1]$ and $R[x_2, x'_2]$ are both non-empty.
        We prove via a contradiction that $x_1 \leq x_2$.
        The proof that $x'_1 \leq x'_2$ is symmetric.
        
        Suppose that $x_1 > x_2$.
        Let $b \in \NN(r)$.
        We naturally have that $R(x_1)$ comes before $r$, and thus $R(x_2)$ comes before $r$.
        The subcurve $R[x_2, x_1]$ and the point $b'$ therefore lie on opposite sides of $\pi(r, b)$.
        It therefore follows from~\cref{lem:shortcutting} that $d(\hat{r}, b) \leq d(\hat{r}, b') \leq \delta$ for all $\hat{r} \in R[x_2, x_1]$.
        By maximality of the leaf of $F_{r, b}$, we thus must have that $R[x_2, x_1]$ is part of the leaf, contradicting the fact that $R[x_1, x'_1]$ is the leaf.
    \end{proof}
    
    \noindent
    The monotonicity of the nearest neighbor fans, together with the fact that each point on $\hat{B}$ corresponds to such a fan, ensures that we can determine such an exit in logarithmic time (see~\cref{lem:advancing_through_fans}).
    We make use of the following data structure that reports transit exits:

    \begin{lemma}
    \label{lem:critical_exit_DS}
        Given the exit interval of $H_\near$, we can report the at most three transit exits on a horizontal line segment $[i, i+1] \times \{y\}$, for any integer $1 \leq i < n$, in $\bigO(\log nm)$ time, after $\bigO(n+m)$ preprocessing time.
    \end{lemma}
    \begin{proof}
        If $(i, y)$ or $(i+1, y)$ is an exit of $H_\near$, then it is also naturally a transit exit, as $R(i)$ and $R(i+1)$ correspond to vertices.
        The third transit exit is of the form $(x, y)$ where $R(x)$ is locally closest to $B(y)$.
        To compute $R(x)$, we make use of the data structure of Cook and Wenk~\cite[Lemma~3.2]{cook10geodesic_frechet}.
        This data structure takes $\bigO(n+m)$ time to construct, and given a query consisting of the edge $R[i, i+1]$ and the point $B(y)$, reports the minimum distance between $B(y)$ and a point on $R[i, i+1]$ in $\bigO(\log nm)$ time.
        Additionally, given this distance, it reports the point achieving this distance in $\bigO(\log nm)$ time.
        This point is locally closest to $B(y)$.
    \end{proof}

    \begin{lemma}
    \label{lem:advancing_through_fans}
        Suppose $\delta \geq \delta_H$.
        Let $H_\near \in \calH$ be a slab corresponding to a subcurve of $B$ with only near points.
        Given the exit interval of $H_\near$, together with a $\delta$-safe transit entrance, we can compute a $\delta$-safe transit exit in $\bigO(\log nm)$ time, after $\bigO(n+m)$ preprocessing time.
    \end{lemma}
    \begin{proof}
        Let $B[y, y']$ be the subcurve of $B$ corresponding to $H_\near$.
        Because $B[y, y']$ contains only near points, each point $B(\hat{y}) \in B[y, y']$ is the apex of some nearest neighbor fan $F_{r, B(\hat{y})}(\delta)$, which in turn corresponds to a horizontal line segment $[x, x'] \times \{\hat{y}\}$ inside $\F_\delta(R, B)$.
        By~\cref{lem:monotone_leaves}, the union of these line segments is $x$- and $y$-monotone.
        Furthermore, observe that this union is connected.
        Indeed, the assumption that $\delta \geq \delta_H$ implies that the distance between any point on $R$ and their nearest neighbor is at most $\delta$, thus ensuring that each nearest neighbor fan has a non-empty leaf that corresponds to a non-empty line segment in the parameter space.

        Let $p_\enter$ be a $\delta$-safe transit entrance of $H_\near$ and let $q_\exit$ be the leftmost transit exit of $H_\near$ that lies to the right of $p_\enter$.
        The monotonicity and connectedness of the region corresponding to the nearest neighbor fans, together with the fact that it lies inside the $\delta$-free space, implies that $q_\exit$ is the leftmost transit exit of $H_\near$ that is $\delta$-reachable from $p_\enter$.
        Since all transit exits to the right of $q_\exit$ are $\delta$-reachable from $q_\exit$, the exit $q_\exit$ is $\delta$-safe.

        Let $[i, i+1] \times \{y'\}$ be the leftmost line segment of this form that has a non-empty intersection with the exit interval of $H_\near$, and where $(i+1, y')$ lies to the right of $p_\enter$.
        The point $q_\exit$ is the leftmost transit exit on this line segment that lies to the right of $p_\enter$.
        We compute $q_\exit$ in $\bigO(\log nm)$ time with the data structure of~\cref{lem:critical_exit_DS}.
    \end{proof}

\subsection{Slabs of far points}
\label{sub:far_points}

    Next we give an algorithm for computing a $\delta$-safe transit exit of a given slab $H_\far \in \calH$ that corresponds to a subcurve of $B$ with only far points on its interior.
    Our algorithm is approximate: given $\eps > 0$, it computes an \emph{$(\eps, \delta)$-safe} transit exit (if one exists).
    This is a $(1+\eps)\delta$-reachable point from which $(n, m)$ is $\delta$-reachable.
    Such a transit exit behaves like a $\delta$-safe transit exit for the purpose of iteratively constructing a matching.

    To compute an $(\eps, \delta)$-safe transit exit, we make use of an approximate decision algorithm that uses the fact that $\hat{B}$ has only far points on its interior.
    We present this algorithm in~\cref{subsub:decision_far_points}.
    In~\cref{subsub:approximate_reachable_exits} we then apply this approximate decision algorithm in a search procedure, where we search over the $\bigO(n)$ transit exits to compute an $(\eps, \delta)$-safe one.

\subsubsection{Approximate decision algorithm for far points}
\label{subsub:decision_far_points}

    Let $\hat{B}$ be the subcurve of $B$ that corresponds to $H_\far$, so its interior has only far points.
    Let $\hat{R}$ be an arbitrary subcurve of $R$, for which we seek to approximately determine whether $\dF(\hat{R}, \hat{B}) \leq \delta$.
    We report either that $\dF(\hat{R}, \hat{B}) \leq (1+\eps)\delta$, or that $\dF(\hat{R}, \hat{B}) > \delta$.
    
    For our algorithm, we first discretize the space of geodesics between points on $\hat{R}$ and points on $\hat{B}$, by grouping the geodesics into few ($\bigO(1/\eps)$) groups and rerouting the geodesics in a group through some representative point.
    Let $\hat{b}$ and $\hat{b}'$ be the first, respectively last, endpoints of $\hat{B}$.
    There is a point $r$ on $R$ with $\NN(r) = \{\hat{b}, \hat{b}'\}$.
    We observe that this implies the geodesic $\pi(\hat{b}, \hat{b}')$ that connects $\hat{b}$ to $\hat{b}'$ is short with respect to the \f distance, and thus with respect to any relevant value for $\delta$:

    \begin{lemma}
        $d(\hat{b}, \hat{b}') \leq 2\dF(\hat{R}, \hat{B})$.
    \end{lemma}
    \begin{proof}
        Let $r \in R$ be a point with $\NN(r) = \{\hat{b}, \hat{b}'\}$.
        If $r \in \hat{R}$, then due to the properties of nearest neighbors, we have $d(r, \hat{b}) \leq \dF(\hat{R}, \hat{B})$ and $d(r, \hat{b}') \leq \dF(\hat{R}, \hat{B})$.
        The claim then follows from the triangle inequality.

        If $r \notin \hat{R}$, then the subcurve $\hat{R}$ either comes before $r$ along $R$, or after.
        Suppose without loss of generality that $\hat{R}$ comes before $r$.
        Let $\hat{r}$ be the last endpoint of $\hat{R}$.
        By definition of \f distance, we have $d(\hat{r}, \hat{b}') \leq \dF(\hat{R}, \hat{B})$.
        Additionally, it follows from~\cref{lem:shortcutting} that $d(\hat{r}, \hat{b}) \leq d(\hat{r}, \hat{b}') \leq \dF(\hat{R}, \hat{B})$.
        The claim again follows from the triangle inequality.
    \end{proof}
    
    \noindent
    We assume for the remainder that $d(\hat{b}, \hat{b}') \leq 2\delta$; if this is not the case, we immediately report that $\dF(\hat{R}, \hat{B}) > \delta$.
    This assumption means that the geodesic $\pi(\hat{b}, \hat{b}')$ is a short separator between $\hat{R}$ and $\hat{B}$.
    That is, any geodesic between a point on $\hat{R}$ and a point on $\hat{B}$ intersects $\pi(\hat{b}, \hat{b}')$.
    For clarity, we denote by $\pi_\sep$ the separator $\pi(\hat{b}, \hat{b}')$.
    We use the short separator to formulate the (exact) reachability problem as $\bigO(1/\eps)$ subproblems involving one-dimensional curves.
    This is where we incur a small approximation error.

    \begin{figure}
        \centering
        \includegraphics[page=2]{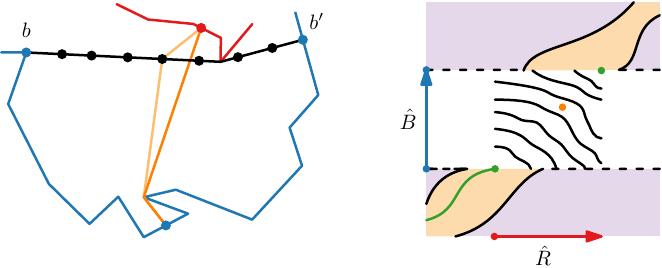}
        \caption{(left) Snapping a geodesic (orange) to an anchor.
        (right) The eight regions in the parameter space of $\hat{R}$ and $\hat{B}$ corresponding to the first eight (out of nine) anchors.
        The orange geodesic lies in region $\calR_5$.}
        \label{fig:far_points}
    \end{figure}
    
    We discretize $\pi_\sep$ with $K+1 = \bigO(1/\eps)$ points $\hat{b} = a_1, \dots, a_{K+1} = \hat{b}'$, which we call \emph{anchors}, and ensure that consecutive anchors at most a distance $\eps \delta$ apart, see~\cref{fig:far_points}~(left).
    We assume that no anchor coincides with a vertex of $\hat{R} \cup \hat{B}$ (except for $a_1$ and $a_{K+1}$).
    
    We route geodesics between points on $\hat{R}$ and points on $\hat{B}$ through these anchors.
    Specifically, for points $\hat{r} \in \hat{R}$ and $\hat{b} \in \hat{B}$, if $\pi(\hat{r}, \hat{b})$ intersects $\pi_\sep$ between consecutive anchors $a_k$ and $a_{k+1}$, then we ``snap'' $\pi(\hat{r}, \hat{b})$ to $a_k$; that is, we replace it by the union of $\pi(\hat{r}, a_k)$ and $\pi(a_k, \hat{b})$ (see~\cref{fig:far_points}~(right)).
    This creates a new path between $\hat{r}$ and $\hat{b}$ that goes through $a_k$ and has length at most $d(\hat{r}, \hat{b}) + \eps \delta$.

    We associate an anchor $a_k$ with the points $(x, y)$ in the parameter space for which the geodesic $\pi(\hat{R}(x), \hat{B}(y))$ is snapped to $a_k$.
    These points form a region $\calR_k$ that is connected and monotone in both coordinates (see~\cref{fig:far_points}~(right)).
    We iteratively compute, for each region $\calR_k$, a set of $(1+\eps)\delta$-reachable points on its boundary, such that the decision question can be answered by checking if the last set contains $(|\hat{R}|, |\hat{B}|)$.
    We later formulate these subproblem in terms of pairs of one-dimensional curves that are separated by a point.

    We first discretize the problem.
    For this, we identify sets of points on the shared boundaries between the pairs of adjacent regions, such that there exists a $\delta$-matching that enters and exits each region through such a set.
    Specifically, for $k = 2, \dots, K$, we let the set $S_k$ contain those points $(x, y) \in \calR_{k-1} \cap \calR_k$ for which one of $R(x)$ and $B(y)$ is a vertex or locally closest to $a_k$ (so perturbing the point infinitesimally along its curve increases its distance to $a_k$).
    We set $S_1 = \{(1, 1)\}$.
    In the following lemmas, we show that these sets suit our needs and are efficiently constructable.
    
    \begin{figure}
        \centering
        \includegraphics{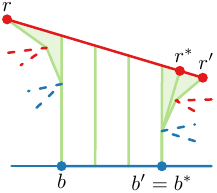}
        \caption{The construction in~\cref{lem:restricted_matching}.
        Non-dashed segments are $\overline{r_i r_{i+1}}$ and $\overline{b_j b_{j+1}}$.}
        \label{fig:restricted_matching}
    \end{figure}

    \begin{lemma}
    \label{lem:restricted_matching}
        If there exists a $\delta$-matching between $\hat{R}$ and $\hat{B}$, then there exists one that goes through a point in $S_k$ for every anchor $a_k$.
    \end{lemma}
    \begin{proof}
        Let $(f, g)$ be a $\delta$-matching between $R$ and $B$.
        Based on $(f, g)$, we construct a new $\delta$-matching $(f', g')$ that satisfies the claim.

        For each vertex $r_i$ of $R$, if $(f, g)$ matches it to a point interior to an edge $\overline{b_j b_{j+1}}$ of $B$, or to the vertex $b_j$, then we let $(f', g')$ match $r_i$ to the point on $\overline{b_j b_{j+1}}$ closest to it.
        This point is either $b_j$, $b_{j+1}$, or it is locally closest to $r_i$.
        Symmetrically, for each vertex $b_j$ of $B$, if $(f, g)$ matches it to a point interior to an edge $\overline{r_i r_{i+1}}$ of $R$, or to the vertex $r_{i+1}$, then we let $(f', g')$ match $b_j$ to the point on $\overline{r_i r_{i+1}}$ closest to it.
        This point is either a vertex or locally closest to $b_j$.

        Consider two maximal subsegments $\overline{r r'}$ and $\overline{b b'}$ of $R$ and $B$ where currently, $r$ is matched to $b$ and $r'$ is matched to $b'$.
        See~\cref{fig:restricted_matching} for an illustration of the following construction.
        Let $\overline{r r'} \subseteq \overline{r_i r_{i+1}}$ and $\overline{b b'} \subseteq \overline{b_j b_{j+1}}$.
        We have $r = r_i$ or $b = b_j$, as well as $r' = r_{i+1}$ or $b' = b_{j+1}$.

        Let $r^* \in \overline{r r'}$ and $b^* \in \overline{b b'}$ minimize the geodesic distance $d(r^*, b^*)$ between them.
        It is clear that $r^*$ and $b^*$ are both either vertices or locally closest to the other point.
        We let $(f', g')$ match $r^*$ to $b^*$.
        Since $(f, g)$ originally matched $r^*$ and $b^*$ to points on $\overline{r_i r_{i+1}}$ and $\overline{b_j b_{j+1}}$, respectively, we have $d(r^*, b^*) \leq \delta$.
        Next we define the part of $(f', g')$ that matches $\overline{r r^*}$ to $\overline{b b^*}$.

        Suppose $r = r_i$ and let $\tilde{b} \in \overline{b b^*}$ be the point closest to $r$.
        We let $(f', g')$ match $r$ to $\overline{b \tilde{b}}$, and match each point on $\overline{r r^*}$ to its closest point on $\overline{\tilde{b} b^*}$.
        This is a proper matching, as the closest point on a segment moves continuously along the segment as we move continuously along $R$.

        The cost of matching $r$ to $\overline{b \tilde{b}}$ is at most $d(r, b) \leq \delta$, since the maximum distance from $r$ to $\overline{b \tilde{b}}$ is attained by $b$ or $\tilde{b}$~\cite{cook10geodesic_frechet}.
        For the cost of matching $\overline{r r^*}$ to $\overline{\tilde{b} b^*}$, observe that the point on $\overline{\tilde{b} b^*}$ closest to a point $\tilde{r} \in \overline{r r^*}$ is also the point on $\overline{b_j b_{j+1}}$ closest to it.
        This is due to $\tilde{b}$ being closest to $r$ among the points on $\overline{b_j b_{j+1}}$ and $b^*$ being locally closest to $r^*$, which means it is closest to $r^*$ among the points on $\overline{b_j b_{j+1}}$.
        It follows that the cost of matching $\overline{r r^*}$ to $\overline{\tilde{b} b^*}$ is at most $\delta$, since $(f, g)$ matches $\overline{r r^*}$ to a subset of $\overline{b_j b_{j+1}}$ and $(f', g')$ matches each point on $\overline{r r^*}$ to its closest point on $\overline{b_j b_{j+1}}$.

        We define a symmetric matching of cost at most $\delta$ between $\overline{r r^*}$ and $\overline{b b^*}$ when $b = b_j$.
        Also, we symmetrically define a matching of cost at most $\delta$ between $\overline{r^* r'}$ and $\overline{b^* b'}$.

        The matching $(f', g')$ has cost at most $\delta$.
        Also, for any pair of matched points $R(x)$ and $B(y)$ where $\pi(R(x), B(y))$ goes through an anchor $a_k$, we have that one of $R(x)$ and $B(y)$ is a vertex, or locally closest to the other point.
        In the latter case, the point is naturally locally closest to $a_k$ as well.
        Thus $(x, y) \in S_k$.        
    \end{proof}

    \begin{lemma}
        Each set $S_k$ contains $\bigO(|\hat{R}|+|\hat{B}|)$ points and can be constructed in $\bigO((|\hat{R}|+|\hat{B}|) \log nm)$ time, after $\bigO(n+m)$ preprocessing time.
    \end{lemma}
    \begin{proof}
        We first bound the number of points in the set $S_k$.
        By our general position assumptions, there are only $|\hat{R}|+|\hat{B}|$ geodesics $\pi(\hat{r}, \hat{b})$ through $a_k$ that have a vertex of $\hat{R}$ or $\hat{B}$ as an endpoint.
        Additionally, observe that for any geodesic $\pi(\hat{r}, \hat{b})$ through $a_k$,
        where $\hat{r} \in \overline{r_i r_{i+1}}$ is locally closest to $\hat{b}$, the point $\hat{r}$ is also closest to $a_k$ among the points on $\overline{r_i r_{i+1}}$.
        Symmetrically, if $\hat{b} \in \overline{b_j b_{j+1}}$ is locally closest to $\hat{r}$, the point $\hat{b}$ is also closest to $a_k$ among the points on $\overline{b_j b_{j+1}}$.
        This gives $|\hat{R}|+|\hat{B}|-2$ geodesics through $a_k$ of this form.
        The set $S_k$ was defined as the set of points in the parameter space corresponding to geodesics of the above three forms.
        Thus $|S_k| = \bigO(|\hat{R}|+|\hat{B}|)$.

        To construct $S_k$, we make use of three data structures.
        We preprocess $P$ into the data structure of~\cite{guibas89shortest_paths} for shortest path queries, which allows for computing the first edge of $\pi(p, q)$ given points $p, q \in P$ in $\bigO(\log nm)$ time.
        We also preprocess $P$ into the data structure of~\cite{chazelle94ray_shooting} for ray shooting queries, which allows for computing the first point on the boundary of $P$ hit by a query ray in $\bigO(\log nm)$ time.
        Lastly, we preprocess $P$ into the data structure of~\cite[Lemma~3.2]{cook10geodesic_frechet}, which in particular allows for computing the point on a segment $e$ that is closest to a point $p$ in $\bigO(\log nm)$ time.
        The preprocessing time for each data structure is $\bigO(n+m)$.

        Given an edge $\overline{r_i r_{i+1}}$ of $\hat{R}$, we compute the point $\hat{r} \in \overline{r_i, r_{i+1}}$ closest to $a_k$, as well as the geodesic $\pi(\hat{r}, \hat{b})$ that goes through $a_k$.
        For this, we first compute the point $\hat{r}$ in $\bigO(\log nm)$ time.
        Then we compute the first edge of $\pi(a_k, \hat{r})$ and extend it towards $B'$ by shooting a ray from $a_k$.
        This takes $\bigO(\log nm)$ time.
        By our general position assumption on the anchors, the ray hits only one point before leaving $P$.
        This is the point $\hat{b}$.
        Through the same procedure, we compute the two geodesics through $p$ that start at $r_i$ and $r_{i+1}$, respectively.

        Applying the above procedure to all vertices and edges of $\hat{R}$, and a symmetric procedure to the vertices and edges of $\hat{B}$, we obtain the set of geodesics corresponding to the points in $S_k$.
        The total time spent is $\bigO(\log nm)$ per vertex or edge, with $\bigO(n+m)$ preprocessing time.
        This sums up to $\bigO((|\hat{R}|+|\hat{B}|) \log nm)$ after preprocessing.
    \end{proof}

    \noindent
    Having constructed the sets $S_k$ for all anchors in $\bigO(\frac{1}{\eps} (|\hat{R}|+|\hat{B}|) \log nm)$ time altogether, we move to computing subsets $S^*_k \subseteq S_k$ containing all $\delta$-reachable points, and only points that are $(1+\eps)\delta$-reachable.
    We proceed iteratively, constructing $S^*_{k+1}$ from $S^*_k$.
    For this, observe that for any point $(x, y)$ in the interior of $\calR_k$, the geodesic $\pi(\hat{R}(x), \hat{B}(y))$ was snapped to $a_k$.
    We use this fact to construct a pair of one-dimensional curves that approximately describe the lengths of these geodesics.

    Let $\bar{R}_k \from [1, |\hat{R}|] \to \R$ and $\bar{B}_k \from [1, |\hat{B}|] \to \R$ be one-dimensional curves, where we set $\bar{R}_k(x) = -d(\hat{R}(x), a_k)$ and $\bar{B}_k(y) = d(\hat{B}(y), a_k)$ (note the difference in sign).
    These curves encode the distances between points on $\hat{R}$ and $\hat{B}$ when snapping geodesics to $a_k$.
    That is, $|\bar{R}_k(x) - \bar{B}_k(y)|$ is the length of $\pi(\hat{R}(x), \hat{B}(y))$ after snapping to $a_k$.
    Hence, for any pair of points $p \in S^*_k$ and $q \in S_{k+1}$, we have the following relations:
    \begin{itemize}
        \item If $q$ is $\delta$-reachable from $p$ in the parameter space of $\hat{R}$ and $\hat{B}$, then it is $(1+\eps)\delta$-reachable in the parameter space of $\bar{R}_k$ and $\bar{B}_k$.
        \item Conversely, if $q$ is $(1+\eps)\delta$-reachable from $p$ in the parameter space of $\bar{R}_k$ and $\bar{B}_k$, then it is $(1+\eps)\delta$-reachable in the parameter space of $\hat{R}$ and $\hat{B}$.
    \end{itemize}
    We define $S^*_{k+1}$ as the points in $S_{k+1}$ that are $(1+\eps)\delta$-reachable from points in $S^*_k$, in the the parameter space of $\bar{R}_k$ and $\bar{B}_k$.
    Computing these points is the problem involving one-dimensional curves that we alluded to earlier.

    \begin{lemma}
        Given $S^*_k$, we can compute $S^*_{k+1}$ in $\bigO((|\hat{R}|+|\hat{B}|) \log nm)$ time, after $\bigO(n+m)$ preprocessing time.
    \end{lemma}
    \begin{proof}
        We first construct the curves $\bar{R}_k$ and $\bar{B}_k$.
        For the \f distance, the parameterization of the curves does not matter.
        This means that for $\bar{R}_k$ and $\bar{B}_k$, we need only the set of local minima and maxima.
        The distance function from $a_k$ to an edge of $\hat{R}$ or $\hat{B}$ has only one local minimum and at most two local maxima~\cite[Lemma~2.1]{cook10geodesic_frechet}, which we compute in $\bigO(\log nm)$ time for a given edge, after $\bigO(n+m)$ preprocessing time~\cite[Lemma~2.1]{cook10geodesic_frechet}.
        The total time to construct $\bar{R}_k$ and $\bar{B}_k$ (under some parameterization) is therefore $\bigO((|\hat{R}|+|\hat{B}|) \log nm)$ after preprocessing.
    
        We compute the set $S^*_{k+1} \subseteq S_{k+1}$ of points that are $(1+\eps)\delta$-reachable from a point in $S^*_k$ in the parameter space of $R_k$ and $B_k$.
        We do so with the algorithm we develop in~\cref{sec:separated_1D} (see~\cref{thm:propagating_reachability}).
        This algorithm takes $\bigO((|\hat{R}|+|\hat{B}|) \log nm)$ time.
    \end{proof}

    \noindent
    We apply the above procedure iteratively, computing $S^*_k$ for each anchor $a_k$.
    These sets take a total of $\bigO(\frac{1}{\eps} (|\hat{B}|+|\hat{R}|) \log nm)$ time to construct.
    Afterwards, if $(|\hat{R}|, |\hat{B}|) \in S^*_K$, we report that $\dF(\hat{R}, \hat{B}) \leq (1+\eps)\delta$.
    Otherwise, we report that $\dF(\hat{R}, \hat{B}) > \delta$.
    We obtain:

    \begin{lemma}
    \label{lem:far_points_decider}
        Let $\hat{R}$ be an arbitrary subcurve of $R$, and let $\hat{B}$ be a maximal subcurve of $B$ with only far points on its interior.
        We can decide whether $\dF(\hat{R}, \hat{B}) \leq (1+\eps)\delta$ or $\dF(\hat{R}, \hat{B}) > \delta$ in $\bigO(\frac{1}{\eps} (|\hat{R}|+|\hat{B}|) \log nm)$ time, after $\bigO(n+m)$ preprocessing time.
    \end{lemma}

\subsubsection{Computing a good exit}
\label{subsub:approximate_reachable_exits}

    \begin{figure}[b]
        \centering
        \includegraphics[page=2]{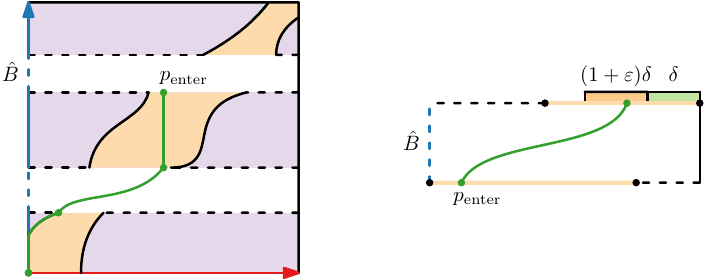}
        \caption{The subproblem of matching far points.
        The exit interval on the right is divided into three regions, based on reachability of points.
        The aim is to compute a $(1+\eps)\delta$-reachable transit exit to the left of all $\delta$-reachable transit exits.}
        \label{fig:matching_far_points}
    \end{figure}

    Recall that we set out to compute an $(\eps, \delta)$-safe transit exit of $H_\far$.
    We assume we are given an $(\eps, \delta)$-safe entrance $p_\enter = (x, y)$.
    Since the entire exit interval of $H_\far$ lies in $\F_\delta(R, B)$, it suffices to compute a transit exit $q_\exit$ that is $(1+\eps)\delta$-reachable from $p_\enter$ and that lies to the left of all transit exits that are $\delta$-reachable from $p_\enter$, see~\cref{fig:matching_far_points}.
    We compute such a transit exit $q_\exit$ through a search procedure, combined with the decision algorithm.

    There are $\bigO(n)$ transit exits of $H_\far$.
    To avoid running the decision algorithm for each of these, we use exponential search.
    The choice for exponential search over, e.g., binary search comes from the fact that the running time of the decision algorithm depends on the location of the transit exit, with transit exits lying further to the right in the exit interval of $H_\far$ needing more time for the decision algorithm.
    Exponential search ensures that we do not consider transit exits that are much more to the right than needed.

    \begin{lemma}
    \label{lem:matching_far_points}
        Let $H_\far \in \calH$ be a slab corresponding to a subcurve $\hat{B}$ of $B$ with only far points on its interior.
        Given an $(\eps, \delta)$-safe transit entrance $p = (x, y)$ of $H_\far$, we can compute an $(\eps, \delta)$-safe transit exit $q = (x', y')$ in $\bigO(\frac{1}{\eps} (|R[x, x']| + |\hat{B}|) \log n \log nm)$ time.
    \end{lemma}
    \begin{proof}
        We search over the edges of $R$.
        For each edge $R[i, i+1]$, we compute the at most three transit exits on the line segment $[i, i+1] \times \{y'\}$ with the data structure of~\cref{lem:critical_exit_DS}, taking $\bigO(\log nm)$ time after $\bigO(n+m)$ preprocessing time.
        We can check whether a transit exit $q_\exit = (x', y')$ is $(1+\eps)\delta$-reachable from $p_\enter = (x, y)$, or not $\delta$-reachable, by applying our decision algorithm~(\cref{lem:matching_far_points}) to the subcurves $R[x, x']$ and $\hat{B} = B[y, y']$.
        If the algorithm reports that $q_\exit$ is $(1+\eps)\delta$-reachable from $p_\enter$, we keep $q_\exit$ in mind and search among transit exits to the left of $q_\exit$.
        Otherwise, we search among the transit exits to the right of~$q_\exit$.
    
        The time spent per candidate exit $q_\exit = (x', y')$ is $\bigO(\frac{1}{\eps} (|R[x, x']| + |\hat{B}|) \log nm)$.
        With exponential search, we consider $\bigO(\log n)$ candidates.
        The total complexity of the subcurves $R[x, x']$ is bounded by $\bigO(|R[x, x^*]|)$, where $(x^*, y')$ is the $(1+\eps)$-approximate $\delta$-safe transit exit we report.
        Thus we get a total time spent of $\bigO(\frac{1}{\eps} (|R[x, x^*]| + |\hat{B}|) \log n \log nm)$.
    \end{proof}

\subsection{The approximate optimization algorithm}
\label{sub:approx_optimization}

    We combine the algorithms of~\cref{sub:near_points,sub:far_points} to obtain a $(1+\eps)$-approximate decision algorithm, which we then turn into an algorithm that computes a $(1+\eps)$-approximation of the geodesic \f distance between $R$ and $B$.
    Given $\delta \geq 0$ and $\eps > 0$, the approximate decision algorithm reports that $\dF(R, B) \leq (1+\eps)\delta$ or $\dF(R, B) > \delta$.

    Let $\delta_H$ be the geodesic Hausdorff distance between $R$ and $B$.
    This distance, which is the maximum distance between a point on $R \cup B$ to its closest point on the other curve, is a natural lower bound on the geodesic \f distance.
    If $\delta < \delta_H$, we therefore immediately return that $\dF(R, B) > \delta$.
    We can compute $\delta_H$ in $\bigO((n+m) \log nm)$ time~\cite{cook10geodesic_frechet}.

    Next suppose $\delta \geq \delta_H$.
    For our approximate decision algorithm, we first compute the partition $\calH$ and the entrance and exit intervals of each of its slabs.
    By \cref{lem:constructing_partition}, this takes $\bigO((n+m) \log nm)$ time.
    We iterate over the $\bigO(m)$ slabs of $\calH$ from bottom to top.
    Once we consider a slab $H \in \calH$, we have computed an $(\eps, \delta)$-safe transit entrance $p_\enter = (x, y)$ (except if $H$ is the bottom slab, in which case we set $p_\enter = (1, 1)$).
    
    Let $\hat{B}$ be the subcurve corresponding to $H$.
    If $\hat{B}$ contains only near points, we compute a $(\eps, \delta)$-safe transit exit $q_\exit = (x', y')$ of $H$ in $\bigO(\log nm)$ time with the algorithm of~\cref{sub:near_points}.
    Otherwise, we use the algorithm of~\cref{sub:far_points}, which takes $\bigO(\frac{1}{\eps} (|R[x, x']| + |\hat{B}|) \log n \log nm)$ time.
    Both algorithms require $\bigO(n+m)$ preprocessing time.
    Taken over all slabs in $\calH$, the total complexity of the subcurves $\hat{B}$ is $\bigO(m)$.
    This gives the following result:

    \begin{theorem}
    \label{thm:approx_decision_algorithm}
        For any $\eps > 0$, there is a $(1+\eps)$-approximate decision algorithm running in $\bigO(\frac{1}{\eps} (n+m) \log n \log nm)$ time.
    \end{theorem}

    \noindent
    We turn the decision algorithm into an approximate optimization algorithm with a simple binary search.
    For this, we show that the geodesic \f distance is not much greater than $\delta_H$.
    This gives an accurate ``guess''  of the actual geodesic \f distance.

    \begin{lemma}
    \label{lem:close_to_hausdorff}
        $\delta_H \leq \dF(R, B) \leq 3\delta_H$.
    \end{lemma}
    \begin{proof}
        Recall that the geodesic Hausdorff distance $\delta_H$ is the maximum distance between a point on $R \cup B$ and its closest point on the other curve.
        It follows directly from this definition that $\delta_H \leq \dF(R, B)$.
        
        Next consider a point $r \in R$ and let $b, b' \in \NN(r)$.
        Take a point $\hat{b}$ between $b$ and $b'$ along $B$.
        There is a point $\hat{r} \in R$ with $d(\hat{r}, \hat{b}) \leq \delta_H$.
        The points $\hat{r}$ and $\hat{b}$ must lie on opposite sides of one of $\pi(r, b)$ and $\pi(r, b')$.
        Hence~\cref{lem:shortcutting} implies that $d(\hat{r}, b) \leq \delta_H$ or $d(\hat{r}, b') \leq \delta_H$.
        Since $d(r, b) \leq \delta_H$ and $d(r, b') \leq \delta_H$, we obtain from the triangle inequality that $d(r, \hat{b}) \leq 3\delta_H$.

        We construct a $(3\delta)$-matching between $R$ and $B$ by matching each point $r \in R$ to the first point in $\NN(r)$.
        If $|\NN(r)| = 2$ and $r$ is the last point with this set of nearest neighbors, we additionally match the entire subcurve of $B$ between the points in $\NN(r)$ to $r$.
        This results in a matching, which has cost at most $3\delta_H$.
    \end{proof}

    \noindent
    For our approximate optimization algorithm, we perform binary search over the values $\delta_H, (1+\eps)\delta_H, \dots, 3\delta_H$ and run our approximate decision algorithm with each encountered parameter.
    This leads to our main result:

   \ourmaintheorem*

\section{Separated one-dimensional curves and propagating reachability}
\label{sec:separated_1D}

    In this section we consider the following problem:
    Let $\bar{R}$ and $\bar{B}$ be two one-dimensional curves with $n$ and $m$ vertices, respectively, where $\bar{R}$ lies left of the point $0$ and $\bar{B}$ right of it.
    We are given a set $S \subseteq \F_\delta(\bar{R}, \bar{B})$ of $\bigO(n+m)$ ``entrances,'' for some $\delta \geq 0$.
    Also, we are given a set $E \subseteq \F_\delta(\bar{R}, \bar{B})$ of $\bigO(n+m)$ ``potential exits.''
    We wish to compute the subset of potential exits that are $\delta$-reachable from an entrance.
    We call this procedure \emph{propagating reachability information} from $S$ to $E$.
    We assume that the points in $S$ and $E$ correspond to pairs of vertices of $\bar{R}$ and $\bar{B}$.
    This assumption can be met by introducing $\bigO(n+m)$ vertices, which does not increase our asymptotic running times.
    Additionally, we may assume that all vertices of $\bar{R}$ and $\bar{B}$ have unique values, for example by a symbolic perturbation.

    The problem of propagating reachability information has already been studied by Bringmann and K\"unnemann~\cite{bringmann17cpacked}.
    In case $S$ lies on the left and bottom sides of the parameter space and $E$ lies on the top and right sides, they give an $\bigO((n+m) \log nm)$ time algorithm.
    We are interested in a more general case however, where $S$ and $E$ may lie anywhere in the parameter space.
    We make heavy use of the concept of \emph{prefix-minima} to develop an algorithm for our more general setting that has the same running time as the one described by Bringmann and K\"unnemann~\cite{bringmann17cpacked}.
    Furthermore, our algorithm is able to actually compute a \f matching between $\bar{R}$ and $\bar{B}$ in linear time (see~\cref{sub:linear_time_frechet}), whereas Bringmann and K\"unnemann require near-linear time for only the decision version.

    %\begin{figure}[b]
    %    \centering
    %    \includegraphics[page=1]{prefix_and_suffix-minima_matchings.pdf}
    %    \caption{(left) A pair of separated, one-dimensional curves $\bar{R}$ and $\bar{B}$, drawn stretched vertically for clarity, together with a prefix-minima matching.
    %    (right) The path in $\F_\delta(\bar{R}, \bar{B})$ corresponding to the matching.
    %    }
    %    \label{fig:prefix-minima_matchings}
    %\end{figure}
    
    As mentioned above, we use prefix-minima extensively for our results in this section. Prefix-minima are those vertices that are closest to the separator $0$ among those points before them on the curves. In~\cref{sub:prefix-minima_matchings} we prove that a \f matching exists that matches subcurves between consecutive prefix-minima to prefix-minima of the other curve (\cref{lem:prefix_minimum_matching}).
    We call these matchings \emph{prefix-minima matchings}.
    This matching will end in a bichromatic closest pair of points (\cref{cor:matching_closest_pair}), and so we can compose the matching with a symmetric matching for the reversed curves.
    
    In~\cref{sub:greedy_paths} we introduce two geometric forests in $\F_\delta(\bar{R}, \bar{B})$, with leaves at $S$, that capture multiple prefix-minima matchings at once. It is based on \emph{horizontal-greedy} and \emph{vertical-greedy} matchings. We show that these forests have linear complexity and can be computed efficiently.
    
    In~\cref{sub:propagating_reachability} we do not only go forward from points in $S$, but also backwards from points in $E$ using \emph{suffix-minima}. Again we have horizontal-greedy and vertical-greedy versions. Intersections between the two prefix-minima forests and the two suffix minima forests show the existence of a $\delta$-free path of the corresponding points in $S$ and $E$, so the problem reduces to a bichromatic intersection algorithm.

\subsection{Prefix-minima matchings}
\label{sub:prefix-minima_matchings}

    We investigate $\delta$-matchings based on \emph{prefix-minima} of the curves.
    We call a vertex $\bar{R}(i)$ a prefix-minimum of $\bar{R}$ if $|\bar{R}(i)| \leq |\bar{R}(x)|$ for all $x \in [1, i]$.
    Prefix-minima of $\bar{B}$ are defined symmetrically.
    Intuitively, prefix-minima are vertices that are closest to $0$ (the separator of $\bar{R}$ and $\bar{B}$) with respect to their corresponding prefix.
    Note that we may extend the definitions to points interior to edges as well, but the restriction to vertices is sufficient for our application.

    The prefix-minima of a curve form a sequence of vertices that monotonically get closer to the separator.
    This leads to the following observation:

    \begin{lemma}
    \label{lem:advancing_to_prefix-minima}
        For any two prefix-minima $\bar{R}(i)$ and $\bar{B}(j)$, we have $\dF(\bar{R}[i, n], \bar{B}[j, m]) \leq \dF(\bar{R}, \bar{B})$.
    \end{lemma}
    \begin{proof}
        Consider a \f matching $(f, g)$ between $\bar{R}$ and $\bar{B}$.
        It matches $\bar{R}(i)$ to some point $\bar{B}(y)$ and matches $\bar{B}(j)$ to some point $\bar{R}(x)$.
        Suppose without loss of generality that $y \geq j$; the other case is symmetric.
        The subcurve $\bar{R}[x, i]$ is matched to the subcurve $\bar{B}[j, y]$.
        By virtue of $\bar{R}(i)$ being a prefix-minimum, it follows that $\dF(\bar{R}(i), \bar{B}[j, y]) \leq \dF(\bar{R}[x, i], \bar{B}[j, y]) \leq \dF(\bar{R}, \bar{B})$.
        Thus composing a matching between $\bar{R}(i)$ and $\bar{B}[j, y]$ with the matching between $\bar{R}[i, n]$ and $\bar{B}[y, m]$ induced by $(f, g)$ gives a matching with cost at most $\dF(\bar{R}, \bar{B})$.
    \end{proof}

    The bichromatic closest pair of points $\bar{R}(i^*)$ and $\bar{B}(j^*)$ is formed by prefix-minima of the curves.
    (This pair of points is unique, by our general position assumption.)
    The points are also prefix-minima of the reversals of the curves.
    By using that the \f distance between two curves is equal to the \f distance between the two reversals of the curves, we obtain the following regarding matchings and bichromatic closest pairs of points:

    \begin{corollary}
    \label{cor:matching_closest_pair}
        There exists a \f matching that matches $\bar{R}(i^*)$ to $\bar{B}(j^*)$.
    \end{corollary}

    \begin{figure}
        \centering
        \includegraphics[page=1]{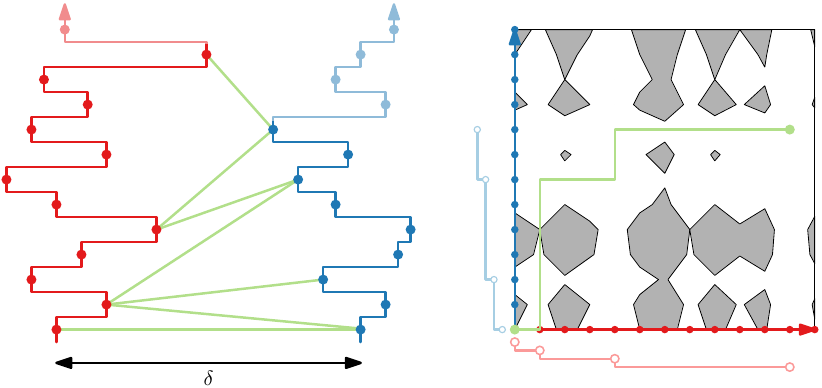}
        \caption{(left) A pair of separated, one-dimensional curves $\bar{R}$ and $\bar{B}$, drawn stretched vertically for clarity.
        A prefix-minima matching, up to the last prefix-minima of the curves, is given in green.
        (right) The path in $\F_\delta(\bar{R}, \bar{B})$ corresponding to the matching.
        }
        \label{fig:prefix-minima_matchings}
    \end{figure}

    A $\delta$-matching $(f, g)$ is called a \emph{prefix-minima $\delta$-matching} if for all $t \in [0, 1]$ at least one of $\bar{R}(f(t))$ and $\bar{B}(g(t))$ is a prefix-minimum.
    Such a matching corresponds to a rectilinear path $\pi$ in $\F_\delta(\bar{R}, \bar{B})$ where for each vertex $(i, j)$ of $\pi$, both $\bar{R}(i)$ and $\bar{B}(j)$ are prefix-minima.
    We call $\pi$ a \emph{prefix-minima $\delta$-matching} as well.
    See~\cref{fig:prefix-minima_matchings} for an illustration.
    We show that there exists a prefix-minima \f matching, up to any pair of prefix-minima:

    \begin{lemma}
    \label{lem:prefix_minimum_matching}
        Let $\bar{R}(i)$ and $\bar{B}(j)$ be prefix-minima of $\bar{R}$ and $\bar{B}$.
        There exists a prefix-minima \f matching between $\bar{R}[1, i]$ and $\bar{B}[1, j]$.
    \end{lemma}
    \begin{proof}
        Let $(f, g)$ be a \f matching between $\bar{R}[1, i]$ and $\bar{B}[1, j]$.
        If $i = 1$ or $j = 1$ then $(f, g)$ is naturally a prefix-minima \f matching.
        We therefore assume $i \geq 2$ and $j \geq 2$, and consider the second prefix-minima $\bar{R}(i')$ and $\bar{B}(j')$ of $\bar{R}$ and $\bar{B}$ (the first being $\bar{R}(1)$ and~$\bar{B}(1)$).

        Let $\bar{R}(\hat{i})$ and $\bar{B}(\hat{j})$ be the points (vertices) on $\bar{R}[1, i']$ and $\bar{B}[1, j']$ furthest from the separator $0$.
        Suppose $(f, g)$ matches $\bar{R}(\hat{i})$ to a point $\bar{B}(\hat{y})$, and matches a point $\bar{R}(\hat{x})$ to $\bar{B}(\hat{j})$.
        We assume that $\hat{x} \leq \hat{i}$; the other case, where $\hat{y} \leq \hat{j}$, is symmetric.

        We have $\hat{i} \leq i'-1$, and hence $\hat{x} \leq i'-1$.
        The subcurve $\bar{R}[1, \hat{x}]$ contains no prefix-minima other than $\bar{R}(1)$, and so
        \[
            |\bar{R}(1) - \bar{B}(y)| \leq |\bar{R}(1) - \bar{B}(\hat{j})| \leq |\bar{R}(\hat{x}) - \bar{B}(\hat{j})| \leq \dF(\bar{R}[1, i], \bar{B}[1, j])
        \]
        for all $y \in [1, j']$.
        The unique matching (up to reparameterization) between $\bar{R}(1)$ and $\bar{B}[1, j']$ is therefore a prefix-minima matching with cost at most $\dF(\bar{R}[1, i], \bar{B}[1, j])$.
        \Cref{lem:advancing_to_prefix-minima} shows that $\dF(\bar{R}, \bar{B}[j', m]) \leq \dF(\bar{R}[1, i], \bar{B}[1, j])$, and so by inductively applying the above construction to $\bar{R}$ and $\bar{B}[j', m]$, we obtain a prefix-minima \f matching between $\bar{R}[1, i]$ and $\bar{B}[1, j]$.
    \end{proof}

\subsection{Greedy paths in the free space}
\label{sub:greedy_paths}

    We wish to construct a set of canonical prefix-minima $\delta$-matchings in the free space from which we can deduce which points in $E$ are reachable.
    Naturally, we want to avoid constructing a path between every point in $S$ and every point in $E$.
    Therefore, we investigate certain classes of prefix-minima $\delta$-matchings that allows us to infer reachability information with just two paths per point in $S$ and two paths per point in $E$.
    Furthermore, these paths have a combined $\bigO(n+m)$ description complexity.

    We first introduce one of the greedy matchings and prove a useful property.
    A \emph{horizontal-greedy $\delta$-matching} $\pi_\hor$ is a prefix-minima $\delta$-matching starting at a point $s = (i, j)$ that satisfies the following property:
    Let $(i', j')$ be a point on $\pi_\hor$ where $\bar{R}(i')$ and $\bar{B}(j')$ are prefix-minima of $\bar{R}[i, n]$ and $\bar{B}[j, m]$.
    If there exists a prefix-minimum $\bar{R}(\hat{i})$ of $\bar{R}[i, n]$ after $\bar{R}(i')$, and the horizontal line segment $[i', \hat{i}] \times \{j'\}$ lies in $\F_\delta(\bar{R}, \bar{B})$, then either $\pi_\hor$ traverses this line segment, or $\pi_\hor$ terminates in $(i', j')$.

    For an entrance $s \in S$, let $\pi_\hor(s)$ be the maximal horizontal-greedy $\delta$-matching.
    See~\cref{fig:horizontal-greedy_forest} for an illustration.
    The path $\pi_\hor(s)$ serves as a canonical prefix-minima $\delta$-matching, in the sense that any point $t$ that is reachable from $s$ by a prefix-minima $\delta$-matching is reachable from a point on $\pi_\hor(s)$ through a single vertical segment:

    \begin{lemma}
    \label{lem:lower_envelope}
        Let $s \in S$ and let $t$ be a point that is reachable by a prefix-minima $\delta$-matching from $s$.
        A point $\hat{t} \in \pi_\hor(s)$ vertically below $t$ exists for which the segment $\overline{\hat{t} t}$ lies in $\F_\delta(\bar{R}, \bar{B})$.
    \end{lemma}
    \begin{proof}
        Let $s = (i, j)$ and $t = (i', j')$.
        Consider a point $(\hat{i}, \hat{j}) \in \pi_\hor(s)$ with $\hat{i} \leq i'$ and $\hat{j} \leq j'$.
        By definition, $\pi_\hor(s)$ is a prefix-minima $\delta$-matching, so $\bar{R}(\hat{i})$ and $\bar{B}(\hat{j})$ are prefix-minima of $\bar{R}[i, n]$ and $\bar{B}[j, m]$, and hence of $\bar{R}[i, i']$ and $\bar{B}[j, j']$.
        By~\cref{lem:advancing_to_prefix-minima}, we have $\dF(\bar{R}[\hat{i}, i'], \bar{B}[\hat{j}, j']) \leq \dF(\bar{R}[i, i'], \bar{B}[j, j']) \leq \delta$.
        So there exists a $\delta$-matching from~$(\hat{i}, \hat{j})$~to~$(i', j')$.

        By the maximality of $\pi_\hor(s)$ and the property that $\pi_\hor(s)$ moves horizontal whenever possible, it follows that $\pi_\hor(s)$ reaches a point $(i', \hat{j})$ with $\hat{j} \leq j'$.
        The existence of a $\delta$-matching from $(i', \hat{j})$ to $(i', j')$ follows from the above.
    \end{proof}

    \begin{figure}
        \centering
        \includegraphics[page=2]{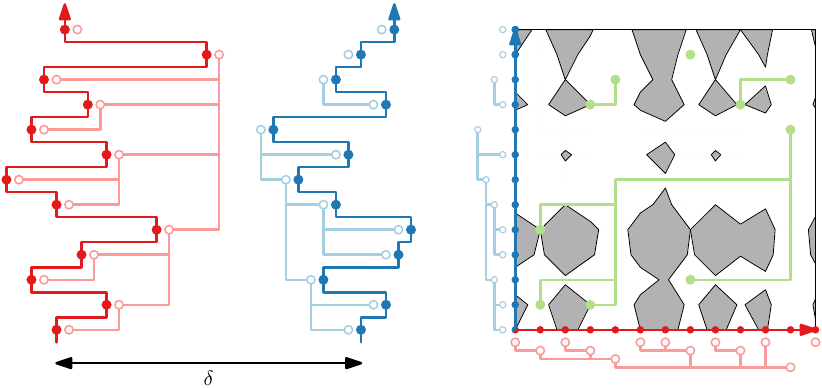}
        \caption{(left) For every vertex, its next prefix-minimum is depicted as its parent in the respective tree. (right) The horizontal-greedy $\delta$-matchings.
        Paths move monotonically to the right and up.
        }
        \label{fig:horizontal-greedy_forest}
    \end{figure}

    \noindent
    A single path $\pi_\hor(s)$ may have $\bigO(n+m)$ complexity.
    We would like to construct the paths for all entrances, but this would result in a combined complexity of $\bigO((n+m)^2)$.
    However, due to the definition of the paths, if two paths $\pi_\hor(s)$ and $\pi_\hor(s')$ have a point $(i, j)$ in common, then the paths are identical from $(i, j)$ onwards.
    Thus, rather than explicitly describing the paths, we instead describe their union.
    Specifically, the set $\bigcup_{s \in S} \pi_\hor(s)$ forms a geometric forest $\T_\hor(S)$ whose leaves are the points in~$S$, see~\cref{fig:horizontal-greedy_forest}.
    In~\cref{sub:forest_complexity} we show that this forest has only $\bigO(n+m)$ complexity, and in~\cref{sub:forest_construction} we give a construction algorithm that takes $\bigO((n+m) \log nm)$ time.

\subsection{Complexity of the forest}
\label{sub:forest_complexity}

    The forest $\T_\hor(S)$ is naturally equal to the union $\bigcup_{\pi \in \Pi} \pi$ of a set of $|S|$ horizontal-greedy $\delta$-matching $\Pi$ with interior-disjoint images that each start at a point in $S$.
    We analyze the complexity of $\T_\hor(S)$ by bounding the complexity of $\bigcup_{\pi \in \Pi} \pi$.

    For the proofs, we introduce the notation $\calC(\pi)$ to denote the set of integers $i \in \{1, \dots, n\}$ for which a path $\pi$ has a vertical edge on the line $\{i\} \times [1, m]$.
    We use the notation $\calR(\pi)$ for representing the horizontal lines containing a horizontal edge of $\pi$.
    The number of edges $\pi$ has is $|\calC(\pi)| + |\calR(\pi)|$.

    \begin{lemma}
    \label{lem:aligned_turns}
        For any two interior-disjoint horizontal-greedy $\delta$-matching $\pi$ and $\pi'$, we have $|\calC(\pi) \cap \calC(\pi')| \leq 1$ or $|\calR(\pi) \cap \calR(\pi')| \leq 1$.
    \end{lemma}
    \begin{proof}
        If $\calC(\pi) \cap \calC(\pi') = \emptyset$ or $\calR(\pi) \cap \calR(\pi') = \emptyset$, the statement trivially holds.
        We therefore assume that the paths have colinear horizontal edges and colinear vertical edges.

        We assume without loss of generality that $\pi$ lies above $\pi'$, so $\pi$ does not have any points that lie vertically below points on $\pi'$.
        Let $e_\hor = [i_1, i_2] \times \{j\}$ and $e_\ver = \{i\} \times [j_1, j_2]$ be the first edges of $\pi$ that are colinear with a horizontal, respectively vertical, edge of $\pi'$.
        Let $e'_\hor = [i'_1, i'_2] \times \{j\}$ and $e'_\ver = \{i\} \times [j'_1, j'_2]$ be the edges of $\pi'$ that are colinear with $e_\hor$ and $e_\ver$, respectively.
        We distinguish between the order of $e_\hor$ and $e_\ver$ along $\pi$.

        First suppose $e_\hor$ comes before $e_\ver$ along $\pi$.
        Let $e'_\ver = \{i\} \times [j'_1, j'_2]$ be the edge of $\pi'$ that is colinear with $e_\ver$.
        This edge lies vertically below $e_\ver$, so $j'_2 \leq j_1$.
        If $\pi'$ terminates at $(i, j'_2)$, then $\calC(\pi) \cap \calC(\pi') = \{i\}$ and the claim holds.
        Next we show that $\pi'$ must terminate at~$(i, j'_2)$.

        Suppose for sake of contradiction that $\pi'$ has a horizontal edge $[i, i'] \times \{j'\}$.
        We have $j' \leq j_1$.
        By virtue of $\pi'$ being a prefix-minima $\delta$-matching, we obtain that $\bar{B}(j')$ is a prefix-minimum of $\bar{B}[y, j']$ for every point $(x, y)$ on $\pi'$.
        In particular, since $\pi'$ has a horizontal edge that is colinear with $e_\hor$, we have that $\bar{B}(j')$ is a prefix-minimum of $\bar{B}[j, j']$, and thus of~$\bar{B}[j, j'_1]$.

        Additionally, by virtue of $\pi$ being a prefix-minima $\delta$-matching, we obtain that $\bar{B}(j'_1)$ is a prefix-minimum of $\bar{B}[j, j'_1]$.
        Hence $|\bar{B}(j'_1)| \leq |\bar{B}(j')|$, which shows that the horizontal line segment $[i, i'] \times \{j'_1\}$ lies in $\F_\delta(\bar{R}, \bar{B})$.
        However, this means that $\pi$ cannot have $e_\ver$ as an edge, as $\pi$ is horizontal-greedy.
        This gives a contradiction.

        The above proves the statement when $e_\hor$ comes before $e_\ver$ along $\pi$.
        Next we prove the statement when $e_\hor$ comes after $e_\ver$ along $\pi$.

        By virtue of $\pi'$ being a prefix-minimum $\delta$-matching, we have that $\bar{B}(j)$ is a prefix-minimum of $\bar{B}[y, j]$ for every point $(x, y)$ on $\pi'$.
        It follows that for all points $(x, y)$ on $\pi'$ with $x \in [i, i'_1]$, we have $|\bar{R}(x) - \bar{B}(j)| \leq |\bar{R}(x) - \bar{B}(y)| \leq \delta$.
        Hence the horizontal line segment $[\max\{i_1, i\}, i'_1] \times \{j\}$ lies in $\F_\delta(\bar{R}, \bar{B})$.
        Because $e_\hor$ comes after $e_\ver$, we further have $i_1 \geq i$.
        Thus, the horizontal-greedy $\delta$-matching $\pi$ must fully contain the horizontal segment $[i_1, i'_1] \times \{j\}$, or terminate in a point on this segment.

        If $\pi$ reaches the point $(i'_1, j)$, then either $\pi$ or $\pi'$ terminates in this point, since the two paths are interior-disjoint.
        Hence we have $\calR(\pi) \cap \calR(\pi') = \{j\}$, proving the statement.
    \end{proof}

    \begin{lemma}
    \label{lem:forest_complexity}
        The forest $\T_\hor(S)$ has $\bigO(n+m)$ vertices.
    \end{lemma}
    \begin{proof}
        We bound the number of edges of $\T_\hor(S)$.
        The forest has at most $|S| = \bigO(n+m)$ connected components, and since each connected component is a tree, the number of vertices of such a component is exactly one greater than the number of edges.
        Thus the number of vertices is at most $\bigO(n+m)$ greater than the number of edges.

        There exists a collection of $|S|$ horizontal-greedy $\delta$-matchings $\Pi$ that all start at points in $S$ and have interior-disjoint images, for which $\T_\hor(S) = \bigcup_{\pi \in \Pi} \pi$.
        Let $\pi_1, \dots, \pi_{|\Pi|}$ be the $\bigO(n+m)$ paths in $\Pi$, in arbitrary order.
        We write $c_i = |\calC(\pi_i)|$, $r_i = |\calR(\pi_i)|$ and $k_i = |\calC(\pi_i)| + |\calR(\pi_i)|$, and proceed to bound $\sum_i k_i$.
        This quantity is equal to the number of edges of $\T_\hor(S)$.

        By~\cref{lem:aligned_turns}, for all pairs of paths $\pi_i$ and $\pi_j$ we have that $|\calC(\pi_i) \cap \calC(\pi_j)| \leq 1$ or $|\calR(\pi_i) \cap \calR(\pi_j)| \leq 1$.
        Let $x_{i, j} \in \{0, 1\}$ be an indicator variable that is set to $1$ if $|\calC(\pi_i) \cap \calC(\pi_j)| \leq 1$ and $0$ if $|\calR(\pi_i) \cap \calR(\pi_j)| \leq 1$ (with an arbitrary value in $\{0, 1\}$ if both hold).
        We then get the following bounds on $c_i$ and $r_i$:
        \[
            c_i \leq n - \sum_{j \neq i} x_{i, j} \cdot (c_j - 1) \quad \text{and} \quad r_i \leq m - \sum_{j \neq i} (1-x_{i, j}) \cdot (r_j - 1).
        \]
        We naturally have that $|c_j - r_j| \leq 1$ for all paths $\pi_j$, and so $k_j = c_j + r_j \leq 2\min\{c_j, r_j\} + 1$.
        Hence
        \[
            c_i \leq n - \sum_{j \neq i} x_{i, j} \cdot \left( \frac{k_j - 1}{2} - 1 \right) \quad \text{and} \quad r_i \leq m - \sum_{j \neq i} (1-x_{i, j}) \cdot \left( \frac{k_j - 1}{2} - 1 \right),
        \]
        from which it follows that
        \[
            k_i \leq n+m - \sum_{j \neq i} \left( \frac{k_j - 1}{2} - 1 \right) = \bigO(n+m) - \frac{1}{2} \sum_{j \neq i} k_j.
        \]
        We proceed to bound the quantity $\sum_{\pi \in \Pi} (|\calC(\pi)| + |\calR(\pi)|) = \sum_i k_i$, which bounds the total number of edges of the paths in $\Pi$, and thus the number of edges of $\T_\hor(S)$:
        \begin{align*}
            \sum_{i=1}^{|\Pi|} k_i &= \sum_{i=3}^{|\Pi|} k_i + k_1 + k_2 \\
            &\leq \sum_{i=3}^{|\Pi|} k_i + \bigO(n+m) - \frac{1}{2} \sum_{j \neq 1} k_j - \frac{1}{2} \sum_{j \neq 2} k_j\\
            &= \sum_{i=3}^{|\Pi|} k_i + \bigO(n+m) - \frac{1}{2}(k_1 + k_2) - \sum_{i=3}^{|\Pi|} k_i\\
            &= \bigO(n+m).\qedhere
        \end{align*}
    \end{proof}

\subsection{Constructing the forest}
\label{sub:forest_construction}

    We turn to constructing the forest $\T_\hor(S)$.
    For this task, we require a data structure that determines, for a vertex of a maximal horizontal-greedy $\delta$-matching, where its next vertex lies.
    We make use of two auxiliary data structures that store one-dimensional curves $A$.
    The first determines, for a given point $A(x)$ and threshold value $U$, the maximum subcurve $A[x, x']$ on which no point's value exceeds $U$.

    \begin{lemma}
    \label{lem:max_prefix_below_threshold_ds}
        Let $A$ be a one-dimensional curve with $k$ vertices.
        In $\bigO(k \log k)$ time, we can construct a data structure of $\bigO(k)$ size, such that given a point $A(x)$ and a threshold value $U \geq A(x)$, the last point $A(x')$ with $\max_{\hat{x} \in [x, x']} A(\hat{x}) \leq U$ can be reported in $\bigO(\log k)$ time.
    \end{lemma}
    \begin{proof}
        We use a persistent red-black tree, of which we first describe the ephemeral variant.
        Let $T_i$ be a red-black tree storing the vertices of $A[i, k]$ in its leaves, based on their order along $A$.
        The tree has $\bigO(\log k)$ height.
        We augment every node of $T_i$ with the last vertex stored in its subtree that has the minimum value.
        To build the tree $T_{i-1}$ from $T_i$, we insert $A(i-1)$ into $T_i$ by letting it be the leftmost leaf.
        This insertion operation costs $\bigO(\log k)$ time, but only at most two ``rotations'' are used to rebalance the tree~\cite{guibas78balanced_trees}.
        Each rotation affects $\bigO(1)$ nodes of the tree, and the subtrees containing these nodes require updating their associated vertex.
        There are $\bigO(\log k)$ such subtrees and updating them takes $\bigO(\log k)$ time in total.
        Inserting a point therefore takes $\bigO(\log k)$ time.

        To keep representations of all trees $T_i$ in memory, we use persistence~\cite{driscoll89persistence}.
        With the techniques of~\cite{driscoll89persistence} to make the data structure persistent, we may access any tree $T_i$ in $\bigO(\log k)$ time.
        The trees all have $\bigO(\log k)$ height.
        The time taken to construct all trees is $\bigO(k \log k)$.

        Consider a query with a point $A(x)$ and value $U \in \R$.
        Let $e = A[i, i+1]$ be the edge of $A$ containing $A(x)$, picking $i = x$ if $A(x)$ is a vertex with two incident edges.
        We first compute the last point on $e$ whose value does not exceed the threshold $U$.
        If this point is not the second endpoint $A(i+1)$ of $e$, then we report this point as the answer to the query.
        Otherwise, we continue to report the last vertex $A(i')$ after $A(i+1)$ for which $\max_{\hat{x} \in [x, x']} A(\hat{x}) \leq U$.
        The answer to the query is on the edge $A[i', i'+1]$.

        We first access $T_i$.
        We then traverse $T_i$ from root to leaf in the following manner:
        Suppose we are in a node $\mu$ and let its left subtree store the vertices of $A[i_1, i_2]$ and its right subtree the vertices of $A[i_2+1, i_3]$.
        If the left child of $\mu$ is augmented with a value greater than $U$, then $A(\hat{i}) > U$ for some $\hat{i} \in [i_1, i_2]$.
        In this case, we continue the search by going into the left child of $\mu$.
        Otherwise, we remember $i_2$ as a candidate for $i'$ and continue the search by going into the right child of $\mu$.
        In the end, we have $\bigO(\log k)$ candidates for $i'$, and we pick the last index.

        Given $i'$, we report the last point on the edge $A[i', i'+1]$ (or $A(i')$ itself if $i' = k$) whose value does not exceed $U$ as the answer to the query.
        We find $i'$ in $\bigO(\log k)$ time, giving a query time of $\bigO(\log k)$.
    \end{proof}

    We also make use of a \emph{range minimum query} data structure.
    A range minimum query on a subcurve $A[x, x']$ reports the minimum value of the subcurve.

    This value is either $A(x)$, $A(x')$, or the minimum value of a vertex of $A[\lceil x \rceil, \lfloor x' \rfloor]$.
    Hence range minimum queries can be answered in $\bigO(1)$ time after $\bigO(k)$ time preprocessing (see e.g.~\cite{fischer10succinct_RMQ}).
    However, we give an alternative data structure with $\bigO(\log k)$ query time.
    Our data structure additionally allows us to query a given range for the first value below a given threshold.
    This latter type of query is also needed for the construction of $\T_\hor(S)$.
    The data structure has the added benefit of working in the pointer-machine model of computation.

    \begin{lemma}
    \label{lem:RMQ_ds}
        Let $A$ be a one-dimensional curve with $k$ vertices.
        In $\bigO(k)$ time, we can construct a data structure of $\bigO(k)$ size, such that the minimum values of a query subcurve $A[x, x']$ can be reported in $\bigO(\log k)$ time.
        Additionally, given a threshold value $U \in \R$, the first and last points $A(x^*)$ on $A[x, x']$ with $A(x^*) \leq U$ (if they exist) can be reported in $\bigO(\log k)$ time.
    \end{lemma}
    \begin{proof}
        We show how to preprocess $A$ for querying the minimum value of a subcurve $A[x, x']$, as well as the first point $A(x^*)$ on $A[x, x']$ with $A(x^*) \leq U$ for a query threshold value $U \in \R$.
        Preprocessing and querying for the other property is symmetric.

        We store the vertices of $A$ in the leaves of a balanced binary search tree $T$, based on their order along $A$.
        We augment each node of $T$ with the minimum value of the vertices stored in its subtree.
        Constructing a balanced binary search tree $T$ on $A$ takes $\bigO(k)$ time, since the vertices are pre-sorted.
        Augmenting the nodes takes $\bigO(k)$ time in total as well, through a bottom-up traversal of $T$.

        Consider a query with a subcurve $A[x, x']$.
        The minimum value of a point on this subcurve is attained by either $A(x)$, $A(x')$, or a vertex of $A[i, i']$ with $i = \lceil x \rceil$ and $i' = \lfloor x' \rfloor$.
        We query $T$ for the minimum value of a vertex of $A[i, i']$.
        For this, we identify $\bigO(\log k)$ nodes whose subtrees combined store exactly the vertices of $A[i, i']$.
        These nodes store a combined $\bigO(\log k)$ candidate values for the minimum, and we identify the minimum in $\bigO(\log k)$ time.
        Comparing this minimum to $A(x)$ and $A(x')$ gives the minimum of $A[x, x']$.

        Given a threshold value $U \in \R$, the first point $A(x^*)$ of $A[x, x']$ with $A(x^*) \leq U$ can be reported similarly to the minimum of the subcurve.
        If $x$ and $x'$ lie on the same edge of $A$, we report the answer in constant time.
        Next suppose $i = \lceil x \rceil \leq \lfloor x' \rfloor = i'$.

        We start by reporting the first vertex $A(i^*)$ of $A[i, i']$ with $A(i^*) \leq U$ (if it exists).
        For this, we again identify $\bigO(\log k)$ nodes whose subtrees combined store exactly the vertices of $A[i, i']$.
        Each node stores the minimum value of the vertices stored in its subtree, and so the leftmost node $\mu$ storing a value below $U$ contains $A(i^*)$.
        (If no such node $\mu$ exists, then $A(i^*)$ does not exist.)
        To get to $A(i^*)$, we traverse the subtree of $\mu$ to a leaf, by always going into the left subtree if it stores a value below $U$.
        Identifying $\mu$ takes $\bigO(\log k)$ time, and traversing its subtree down to $A(i^*)$ takes an additional $\bigO(\log k)$ time.

        Given $A(i^*)$, the point $A(x^*)$ lies on the edge $A[i^*-1, i^*]$ and we compute it in $\bigO(1)$ time.
        If $i^*$ does not exist, then $A(x^*)$ is equal to either $A(x)$ or $A(x')$, and we report $A(x^*)$ in $\bigO(1)$ time.
    \end{proof}

    Next we give two data structures, one that determines how far we may extend a horizontal-greedy $\delta$-matching horizontally, and one that determines how far we may extend it vertically.

\subparagraph*{Horizontal movement.}
    We preprocess $\bar{R}$ into the data structures of~\cref{lem:max_prefix_below_threshold_ds,lem:RMQ_ds}, taking $\bigO(n \log n)$ time.
    To determine the maximum horizontal movement from a given vertex $(i, j)$ on a horizontal-greedy $\delta$-matching $\pi$, we first report the last vertex $\bar{R}(i')$ after $\bar{R}(i)$ for which $\max_{x \in [i, i']} |\bar{R}(x)| \leq \delta - |\bar{B}(j)|$.
    Since $\bar{R}$ lies completely left of $0$, we have $\max_{x \in [i, i']} |\bar{R}(x)| = \min_{x \in [i, i']} \bar{R}(x)$, and so this vertex can be reported in $\bigO(\log n)$ time with the data structure of~\cref{lem:max_prefix_below_threshold_ds}.

    Next we report the last prefix-minimum $\bar{R}(i^*)$ of $\bar{R}[i, i']$.
    The path $\pi$ may move horizontally from $(i, j)$ to $(i^*, j)$ and not further.
    Observe that $\bar{R}(i^*)$ is the last vertex of $\bar{R}[i, i']$ with $|\bar{R}(i^*)| \leq \min_{x \in [i, i']} |\bar{R}(x)|$.
    We report the value of $\min_{x \in [i, i']} |\bar{R}(x)|$ in $\bigO(\log n)$ time with the data structure of~\cref{lem:RMQ_ds}.
    The vertex $\bar{R}(i^*)$ of $\bar{R}[i, i']$ can then be reported in $\bigO(\log n)$ additional time with the same data structure.

    \begin{lemma}
    \label{lem:horizontal_movement_ds}
        In $\bigO(n \log n)$ time, we can construct a data structure of $\bigO(n)$ size, such that given a vertex $(i, j)$ of a horizontal-greedy $\delta$-matching $\pi$, the maximal horizontal line segment that $\pi$ may use as an edge from $(i, j)$ can be reported in $\bigO(\log n)$ time.
    \end{lemma}

\subparagraph*{Vertical movement.}
    To determine the maximum vertical movement from a given vertex $(i, j)$ on a horizontal-greedy $\delta$-matching $\pi$, we need to determine the first prefix-minimum $\bar{B}(j')$ of $\bar{B}[j, j']$ for which a horizontal-greedy $\delta$-matching needs to move horizontally from $(i, j')$.
    For this, we make use of the following data structure that determines the second prefix-minimum $\bar{R}(i^*)$ of $\bar{R}[i, n]$:

    \begin{lemma}
        In $\bigO(n)$ time, we can construct a data structure of $\bigO(n)$ size, such that given a vertex $\bar{R}(i)$, the second prefix-minimum of $\bar{R}[i, n]$ (if it exists) can be reported in $\bigO(1)$ time.
    \end{lemma}
    \begin{proof}
        We use the algorithm of Berkman~\etal~\cite{berkman93all_nearest_smaller_values} to compute, for every vertex $\bar{R}(i)$ of $\bar{R}$, the first vertex $\bar{R}(i^*)$ after $\bar{R}(i)$ with $|\bar{R}(i^*)| \leq |\bar{R}(i)|$ (if it exists).
        Naturally, $\bar{R}(i^*)$ is the second prefix-minimum of $R[i, n]$.
        Their algorithm takes $\bigO(n)$ time.
        Annotating the vertices of $\bar{R}$ with their respective second prefix-minima gives a $\bigO(n)$-size data structure with $\bigO(1)$ query time.
    \end{proof}

    The first prefix-minimum $\bar{B}(j')$ of $\bar{B}[j, j']$ for which a horizontal-greedy $\delta$-matching needs to move horizontally from $(i, j')$, is the first prefix-minimum with $\dF(\bar{R}[i, i^*], \bar{B}(j')) \leq \delta$.
    Observe that $\bar{B}(j')$ is not only the first prefix-minimum of $\bar{B}[j, m]$ with $\dF(\bar{R}[i, i^*], \bar{B}(j')) \leq \delta$, it is also the first vertex of $\bar{B}[j, m]$ with this property.

    We preprocess $\bar{B}$ into the data structures of~\cref{lem:max_prefix_below_threshold_ds,lem:RMQ_ds}, taking $\bigO(m \log m)$ time.
    We additionally preprocess $\bar{R}$ into the data structure of~\cref{lem:RMQ_ds}, taking $\bigO(n \log n)$ time.

    We first compute $\max_{x \in [i, i^*]} |\bar{R}(x)| = \min_{x \in [i, i^*]} \bar{R}(x)$ with the data structure of~\cref{lem:RMQ_ds}, taking $\bigO(\log n)$ time.
    We then report $\bar{B}(j')$ as the first vertex of $\bar{B}[j, m]$ with $|\bar{B}(j')| \leq \delta - \max_{x \in [i, i^*]} |\bar{R}(x)|$.
    This takes $\bigO(\log m)$ time.

    To determine the maximum vertical movement of $\pi$, we need to compute the maximum vertical line segment $\{i\} \times [j, j^*] \subseteq \F_\delta(\bar{R}, \bar{B})$ for which $\bar{B}(j^*)$ is a prefix-minimum of $\bar{B}[j, j']$.
    We query the data structure of~\cref{lem:max_prefix_below_threshold_ds} for the last vertex $\bar{B}(\hat{j})$ for which $\max_{y \in [j, \hat{j}]} |\bar{B}(y)| \leq \delta - |\bar{R}(i)|$.
    This takes $\bigO(\log m)$ time.
    The vertex $\bar{B}(j^*)$ is then the last prefix-minimum of $\bar{B}[j, \hat{j}]$, which we report in $\bigO(\log m)$ time with the data structure of~\cref{lem:RMQ_ds}, by first computing the minimum value of $\bar{B}[j, \hat{j}]$ and then the vertex that attains this value.

    \begin{lemma}
    \label{lem:vertical_movement_ds}
        In $\bigO((n+m) \log nm)$ time, we can construct a data structure of $\bigO(n+m)$ size, such that given a vertex $(i, j)$ of a horizontal-greedy $\delta$-matching $\pi$, the maximal vertical line segment that $\pi$ may use as an edge from $(i, j)$ can be reported in $\bigO(\log nm)$ time.
    \end{lemma}

\subparagraph*{Completing the construction.}
    We proceed to iteratively construct the forest $\T_\hor(S)$.
    Let $S' \subseteq S$ and suppose we have the forest $\T_\hor(S')$.
    Initially, $\T_\hor(\emptyset)$ is the empty forest.
    We show how to construct $\T_\hor(S' \cup \{s\})$ for a point $s \in S \setminus S'$.

    We assume $\T_\hor(S')$ is represented as a geometric graph.
    Further, we assume that the $\bigO(n+m)$ vertices of $\T_\hor(S')$ are stored in a red-black tree, based on the lexicographical ordering of the endpoints.
    This allows us to query whether a given point is a vertex of $\T_\hor(S')$ in $\bigO(\log nm)$ time, and also allows us to insert new vertices in $\bigO(\log nm)$ time each.

    We use the data structures of~\cref{lem:horizontal_movement_ds,lem:vertical_movement_ds}.
    These allow us compute the edge of $\pi_\hor(s)$ after a given vertex in $\bigO(\log nm)$ time.
    The preprocessing for the data structures is $\bigO((n+m) \log nm)$.

    We construct the prefix $\pi$ of $\pi_\hor(s)$ up to the first vertex of $\pi_\hor(s)$ that is a vertex of $\T_\hor(S')$ (or until the last vertex of $\pi_\hor(s)$ if no such vertex exists).
    This takes $\bigO(|\pi| \log nm)$ time, where $|\pi|$ is the number of vertices of $\pi$.
    Recall that if two maximal horizontal-greedy $\delta$-matchings $\pi_\hor(s)$ and $\pi_\hor(s')$ have a point $(x, y)$ in common, then the paths are identical from $(x, y)$ onwards.
    Thus the remainder of $\pi_\hor(s)$ is a path in $\T_\hor(S')$.

    We add all vertices and edges of $\pi$, except the last two vertices $p$ and $q$ and the last edge $e = \overline{pq}$, to $\T_\hor(S')$.
    If $\T_\hor(S')$ does not have a vertex at point $q$ already, then $q$ does not lie anywhere on $\T_\hor(S')$, not even interior to an edge.
    Hence $\pi$ is completely disjoint from $\T_\hor(S')$, and we add $p$ and $q$ as vertices to the forest, and $e$ as an edge.
    If $\T_\hor(S')$ does have a vertex $\mu$ at point $q$, then the edge $e$ may overlap with an edge of $\T_\hor(S')$.
    In this case, we retrieve the edge $e_\mu$ of $\T_\hor(S')$ that overlaps with $e$ (if it exists), by identifying the edges incident to $\mu$.
    If $p$ lies on $e_\mu$, we subdivide $e_\mu$ by adding a vertex at $p$.
    If $p$ does not lie on $e_\mu$, then the endpoint $q_\mu$ of $e_\mu$ other than $q$ lies on the interior of $e$.
    We add an edge from $p$ to $q_\mu$.

    The above construction updates $\T_\hor(S')$ into the forest $\T_\hor(S' \cup \{s\})$ in $\bigO(|\pi| \log nm)$ time.
    Inserting all $\bigO(|\pi|)$ newly added vertices into the red-black tree takes an additional $\bigO(|\pi| \log nm)$ time.
    It follows from the combined $\bigO(n+m)$ complexity of $\T_\hor(S)$ that constructing $\T_\hor(S)$ in this manner takes $\bigO((n+m) \log nm)$ time.

    \begin{lemma}
    \label{lem:constructing_forest}
        We can construct a geometric graph for $\T_\hor(S)$ in $\bigO((n+m) \log nm)$ time.
    \end{lemma}

\subsection{Propagating reachability}
\label{sub:propagating_reachability}

    Next we give an algorithm for propagating reachability information.
    For the algorithm, we consider three more $\delta$-matchings that are symmetric in definition to the horizontal-greedy $\delta$-matchings.
    The first is the maximal \emph{vertical-greedy} $\delta$-matching $\pi_\ver(s)$, which, as the name suggests, is the maximal prefix-minima $\delta$-matching starting at $s$ that prioritizes vertical movement over horizontal movement.
    The other two require a symmetric definition to prefix-minima, namely \emph{suffix-minima}.
    These are the vertices closest to $0$ compared to the suffix of the curve after the vertex.
    The maximal \emph{reverse} horizontal- and vertical-greedy $\delta$-matchings $\rev{\pi}_\hor(t)$ and $\rev{\pi}_\ver(t)$ are symmetric in definition to the maximal horizontal- and vertical-greedy $\delta$-matching, except that they move backwards, to the left and down, and their vertices correspond to suffix-minima of the curves (see~\cref{fig:suffix-minima_matchings}).

    \begin{figure}
        \centering
        \includegraphics[page=3]{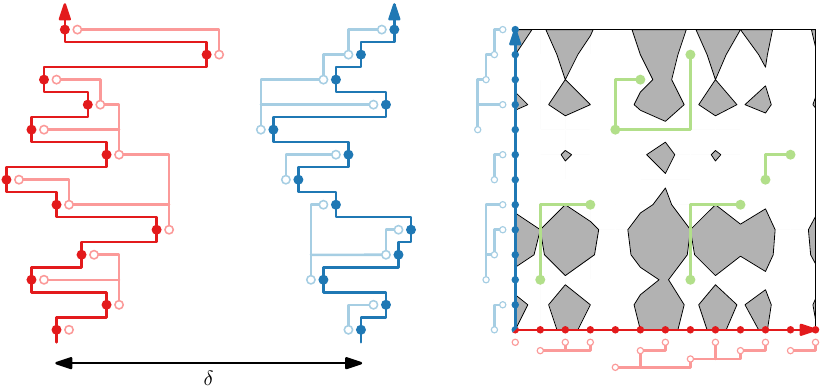}
        \caption{(left) For every vertex, its previous suffix-minimum is shown as its parent in the tree. (right) The reverse horizontal-greedy $\delta$-matchings.
        Paths move monotonically to the left and down.
        }
        \label{fig:suffix-minima_matchings}
    \end{figure}

    Consider a point $s = (i, j) \in S$ and let $t = (i', j') \in E$ be $\delta$-reachable from $s$.
    Let $\bar{R}(i^*)$ and $\bar{B}(j^*)$ form a bichromatic closest pair of $\bar{R}[i, i']$ and $\bar{B}[j, j']$.
    Note that these points are unique, by our general position assumption.
    Recall from~\cref{cor:matching_closest_pair} that $(i^*, j^*)$ is $\delta$-reachable from $s$, and that $t$ is $\delta$-reachable from $(i^*, j^*)$.

    From~\cref{lem:lower_envelope} we have that $\pi_\hor(s)$ has points vertically below $(i^*, j^*)$, and the vertical segment between $\pi_\hor(s)$ and $(i^*, j^*)$ lies in $\F_\delta(\bar{R}, \bar{B})$.
    We extend the property to somewhat predict the movement of $\pi_\hor(s)$ near~$t$:

    \begin{lemma}
    \label{lem:horizontal_greedy_paths}
        Either $\pi_\hor(s)$ terminates in $(i^*, j^*)$, or it contains a point vertically below $t$ or horizontally left of $t$.
    \end{lemma}
    \begin{proof}
        From~\cref{lem:lower_envelope} we obtain that there exists a point $(i^*, \hat{j})$ on $\pi_\hor(s)$ that lies vertically below $(i^*, j^*)$.
        Moreover, the vertical line segment $\{i^*\} \times [\hat{j}, j^*]$ lies in $\F_\delta(\bar{R}, \bar{B})$.
        Because $\bar{R}(i^*)$ and $\bar{B}(j^*)$ form the unique bichromatic closest pair of $\bar{R}[i, i']$ and $\bar{B}[j, j']$, we have that $\bar{R}(i^*)$ and $\bar{B}(j^*)$ are the last prefix-minima of $\bar{R}[i, i']$ and $\bar{B}[j, j']$.
        Hence $\pi_\hor(s)$ has no vertex in the vertical slab $[i^*+1, i'] \times [1, m]$.
        Symmetrically, $\pi_\hor(s)$ has no vertex in the horizontal slab $[1, n] \times [j^*+1, j']$.
        Maximality of $\pi_\hor(s)$ therefore implies that $\pi_\hor(s)$ either moves horizontally from $(i^*, \hat{j})$ past $(i', \hat{j})$, or $\pi_\hor(s)$ moves vertically from $(i^*, \hat{j})$ to $(i^*, j^*)$, where it either terminates or moves further upwards past $(i^*, j')$.
    \end{proof}

    \begin{figure}[b]
        \centering
        \includegraphics{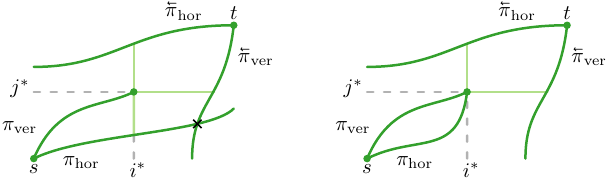}
        \caption{Two possible situations following from~\cref{lem:lower_envelope,lem:horizontal_greedy_paths}.
        The paths starting at $s$ or~$t$ are the four greedy matchings.
        The horizontal and vertical light green segments lie in $\F_\delta(R, B)$.
        On the right, the extensions of $\pi_\hor(s)$ and $\pi_\ver(s)$ respectively intersect $\rev{\pi}_\ver(t)$ and~$\rev{\pi}_\hor(t)$.}
        \label{fig:greedy_paths}
    \end{figure}

    \noindent
    Based on~\cref{lem:lower_envelope,lem:horizontal_greedy_paths} and their symmetric counterparts, $\pi_\hor(s)$ and $\pi_\ver(s)$ satisfy the properties below, and $\rev{\pi}_\hor(t)$ and $\rev{\pi}_\ver(t)$ satisfy symmetric properties, see~\cref{fig:greedy_paths}.
    \begin{itemize}
        \item $\pi_\hor(s)$ has a point vertically below $(i^*, j^*)$, and the vertical segment between $\pi_\hor(s)$ and $(i^*, j^*)$ lies in $\F_\delta(\bar{R}, \bar{B})$.

        \item $\pi_\ver(s)$ has a point horizontally left of $(i^*, j^*)$, and the horizontal segment between $\pi_\ver(s)$ and $(i^*, j^*)$ lies in $\F_\delta(\bar{R}, \bar{B})$.

        \item $\pi_\hor(s)$ and $\pi_\ver(s)$ both either terminate in $(i^*, j^*)$, or contain a point vertically below $t$ or horizontally left of $t$.
    \end{itemize}
\noindent
    These properties mean that either $\pi_\hor(s) \cup \pi_\ver(s)$ intersects $\rev{\pi}_\hor(t) \cup \rev{\pi}_\ver(t)$, or the following extensions do:
    Let $\pi^+_\hor(s)$ be the path obtained by extending $\pi_\hor(s)$ with the maximum horizontal line segment in $\F_\delta(\bar{R}, \bar{B})$ whose left endpoint is the end of $\pi_\hor(s)$.
    Define $\pi^+_\ver(s)$ symmetrically, by extending $\pi_\ver(s)$ with a vertical segment.
    Also define $\rev{\pi}^+_\hor(s)$ and $\rev{\pi}^+_\ver(s)$ analogously.
    By~\cref{lem:lower_envelope}, $\pi^+_\hor(s)$ or $\pi^+_\ver(s)$ \emph{must} intersect $\rev{\pi}^+_\hor(t)$ or $\rev{\pi}^+_\ver(t)$.
    Furthermore, if $\pi^+_\hor(s)$ or $\pi^+_\ver(s)$ intersects $\rev{\pi}^+_\hor(t')$ or $\rev{\pi}^+_\ver(t')$ for some potential exit $t' \in E$, then the bimonotonicity of the paths implies that $t'$ is $\delta$-reachable from $s$.
    Thus:

    \begin{lemma}
        A point $t \in E$ is $\delta$-reachable from a point $s \in S$ if and only if $\pi^+_\hor(s) \cup \pi^+_\ver(s)$ intersects $\rev{\pi}^+_\hor(t') \cup \rev{\pi}^+_\ver(t')$.
    \end{lemma}

    \noindent
    Recall that $\T_\hor(S)$ represents all paths $\pi_\hor(s)$.
    We augment $\T_\hor(S)$ to represent all paths $\pi^+_\hor(s)$.
    For this, we take each root vertex $p$ and compute the maximal horizontal segment $\overline{pq} \subseteq \F_\delta(\bar{R}, \bar{B})$ that has $p$ as its left endpoint.
    We compute this segment in $\bigO(\log n)$ time after $\bigO(n \log n)$ time preprocessing (see~\cref{lem:max_prefix_below_threshold_ds}).
    We then add $q$ as a vertex to $\T_\hor(S)$, and add an edge from $p$ to $q$.

    Let $\T^+_\hor(S)$ be the augmented graph.
    We define the graphs $\T^+_\ver(S)$, $\rev{\T}^+_\hor(E)$ and $\rev{\T}^+_\ver(E)$ analogously.
    The four graphs have a combined complexity of $\bigO(n+m)$ and can be constructed in $\bigO((n+m) \log nm)$ time.
    Our algorithm computes the edges of $\rev{\T}^+_\hor(E)$ and $\rev{\T}^+_\ver(E)$ that intersect an edge of $\T^+_\hor(S)$ or $\T^+_\ver(S)$.
    We do so with a standard sweepline algorithm:

    \begin{lemma}
    \label{lem:reporting_some_intersections}
        Given sets of $n$ ``red'' and $m$ ``blue,'' axis-aligned line segments in $\R^2$, we can report all segments that intersect a segment of the other color in $\bigO((n+m) \log nm)$ time.
    \end{lemma}
    \begin{proof}
        Let $L_R$ be the set of $n$ red segments and let $L_B$ be the set of $m$ blue segments.
        We give an algorithm that reports all red segments that intersect a blue segment.
        Reporting all blue segments that intersect a red segment can be done symmetrically.

        We give a horizontal sweepline algorithm, where we sweep upwards.
        During the sweep, we maintain three structures:
        \begin{enumerate}
            \item The set $L^*_R$ of segments in $L_R$ for which we have swept over an intersection with a segment in $L_B$.
            \item An interval tree~\cite{cormen89introduction_to_algorithms} $T_R$ storing the intersections between segments in $L_R \setminus L^*_R$ and the sweepline (viewing the sweepline as a number line).
            \item An interval tree $T_B$ storing the intersections between segments in $L_B$ and the sweepline.
        \end{enumerate}
        The trees $T_R$ and $T_B$ use $\bigO(|L_R \setminus L^*_R|)$, respectively $\bigO(|L_B|)$, space.
        The trees allow for querying whether a given interval intersects an interval in the tree in time logarithmic in the size of the tree, and allow for reporting all intersected intervals in additional time linear in the output size.
        Furthermore, the trees allow for insertions and deletions in time logarithmic in their size.

        The interval trees change only when the sweepline encounters an endpoint of a segment.
        Moreover, if two segments $e_R \in L_R$ and $e_B \in L_B$ intersect, then they have an intersection point that lies on the same horizontal line as an endpoint of $e_R$ or $e_B$.
        Hence it also suffices to update $L^*_R$ only when the sweepline encounters an endpoint.
        We next discuss how to update the structures.

        Upon encountering an endpoint, we first update the interval trees $T_R$ and $T_B$ by inserting the set of segments whose bottom-left endpoint lies on the sweepline.
        Let $L'_R \subseteq L_R$ and $L'_B \subseteq L_B$ be the sets of newly inserted segments.

        For each segment $e \in L'_R$, we check whether it intersects a line segment in $L_B$ in a point on the sweepline.
        For this, we query the interval tree $T_B$, which reports whether there exists an interval overlapping the interval corresponding to $e$ in $\bigO(\log m)$ time.
        If the query reports affirmative, we insert $e$ into $L^*_R$ and remove it from $L'_R$.
        The total time for this step is $\bigO(|L'_R| \log m)$.

        For each segment $e$ in $L'_B$, we report the line segments in $L_R \setminus L^*_R$ that have an intersection with $e$ on the sweepline.
        Doing so takes $\bigO(\log n + k_e)$ time by querying $T_R$, where $k_e$ is the number of segments reported for $e$.
        Before reporting the intersections of the next segment in $L'_B$, we first add all $k_e$ reported segments to $L^*_R$, remove them from $L'_R$, and remove their corrsponding intervals from $T_R$.
        This ensures that we report every segment at most once.
        Updating the structures takes $\bigO((1+k_e) \log n)$ time.
        Taken over all segments $e \in L'_B$, the total time taken for this step is $\bigO((|L'_B|+\sum_{e \in L'_R} k_e) \log n)$.

        Finally, we remove each segment from $T_R$ and $T_B$ whose top-right endpoint lies on the sweepline, as these are no longer intersected by the sweepline when advancing the sweep.

        Computing the events of the sweepline takes $\bigO((n+m) \log nm)$ time, by sorting the endpoints of the segments by $y$-coordinate.
        Each red, respectively blue, segment inserted and deleted from its respective interval tree exactly once.
        Hence each segment is included in $L'_R$ or $L'_B$ exactly once.
        It follows that the total computation time is $\bigO((n+m) \log nm)$.
    \end{proof}

    \noindent
    Suppose we have computed the set of edges $\mathcal{E}$ of $\rev{\T}^+_\hor(E)$ and $\rev{\T}^+_\ver(E)$ that intersect an edge of $\T^+_\hor(S)$ or $\T^+_\ver(S)$.
    We store $\mathcal{E}$ in a red-black tree, so that we can efficiently retrieve and remove edges from this set.
    Let $e \in \mathcal{E}$ and suppose $e$ is an edge of $\rev{\T}^+_\hor(E)$.
    Let $\mu$ be the top-right vertex of $e$.
    All potential exits of $E$ that are stored in the subtree of $\mu$ are reachable from a point in $S$.
    We traverse the entire subtree of $\mu$, deleting every edge we find from $\mathcal{E}$.
    Every point in $E$ we find is marked as reachable.
    In this manner, we obtain:

    \begin{theorem}
    \label{thm:propagating_reachability}
        Let $\bar{R}$ and $\bar{B}$ be two separated one-dimensional curves with $n$ and $m$ vertices.
        Let $\delta \geq 0$, and let $S, E \subseteq \F_\delta(\bar{R}, \bar{B})$ be sets of $\bigO(n+m)$ points.
        We can compute the set of all points in $E$ that are $\delta$-reachable from points in $S$ in $\bigO((n+m) \log nm)$ time.
    \end{theorem}

\section{Conclusion}
    We studied computing the approximate \f distance of two curves $R$ and $B$ that bound a simple polygon $P$,
    one clockwise and one counterclockwise, whose endpoints meet.
    Our algorithm is approximate, though the only approximate parts are for matching the far points and turning the decision algorithm into an optimization algorithm.
    Doing so exactly and in strongly subquadratic time remains an interesting open problem.
    
    Our algorithm extends to the case where $R$ and $B$ do not cover the complete boundary of the polygon.
    In other words, the start and endpoints of $R$ and $B$ need not coincide. Geodesics between points on $R$ and $B$ must stay inside $P$.
    In this case, $k=|P|$ can be much greater than $n+m-2$, which influences the preprocessing and query times of various data structures we use.
    The running time then becomes: $\bigO(k + \frac{1}{\eps} (n+m \log n) \log k \log \frac{1}{\eps})$.

\bibliographystyle{plainurl}
\bibliography{bibliography}

\newpage
\appendix

\section{Convex polygon}
\label{app:convex}

    Let $R \from [1, n] \to \R^2$ and $B \from [1, m] \to \R^2$ be two curves that bound a convex polygon $P$.
    We assume that $R$ moves clockwise and $B$ counter-clockwise around $P$.
    We give a simple linear-time algorithm for computing the geodesic \f distance between $R$ and $B$.
    For this we show that in this setting, there exists a \f matching of a particular structure, which we call a \emph{maximally-parallel matching}.

    \begin{figure}[b]
        \centering
        \includegraphics[page=2]{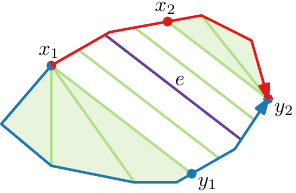}
        \caption{An illustration of the maximally-parallel matching with respect to $\ell$.}
        \label{fig:maximally_parallel}
    \end{figure}

    Consider a line $\ell$.
    Let $R[r, r']$ and $B[b, b']$ be the maximal subcurves for which the lines supporting $\overline{rb}$ and $\overline{r'b'}$ are parallel to $\ell$, and $R[r, r']$ and $B[b, b']$ are contained in the strip bounded by these lines.
    Using the convexity of the curves, it must be that $r = R(1)$ or $b = B(1)$, as well as $r' = R(n)$ or $b' = B(m)$.
    The maximally-parallel matching with respect to $\ell$ matches $R[r, r']$ to $B[b, b']$, such that for every pair of points matched, the line through them is parallel to $\ell$.
    The rest of the matching matches the prefix of $R$ up to $r$ to the prefix of $B$ up to $b$, and matches the suffix of $R$ from $r'$ to the suffix of $B$ from $b'$, where for both parts, one of the subcurves is a single point.
    See~\cref{fig:maximally_parallel} for an illustration.
    We refer to the three parts of the matching as the first ``fan'', the ``parallel'' part, and the last ``fan''.

    In~\cref{lem:maximally-parallel} we show the existence of a maximally-parallel \f matching.
    Moreover, we show that there exists a \f matching that is a maximally-parallel matching with respect to a particular line that proves useful for our construction algorithm.
    Specifically, we show that there exists a pair of parallel lines $\ell_R$ and $\ell_B$ tangent to $P$, with $P$ between them, such that the maximally-parallel matching with respect to the line through the bichromatic closest pair of points $r^* \in \ell_R \cap R$ and $b^* \in \ell_B \cap B$ is a \f matching.
    To prove that such a matching exists, we first prove that there exist parallel tangents that go through points that are matched by a \f matching:

    \begin{figure}[t]
        \centering
        \includegraphics[page=3]{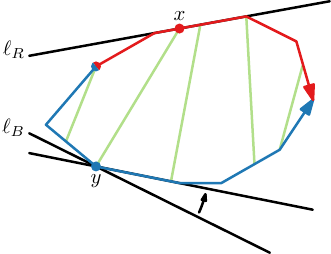}
        \caption{Rotating tangent lines $\ell_R$ and $\ell_B$ around $P$ while the lines touch points $r$ and $b$ matched by the green matching.}
        \label{fig:rotating_calipers}
    \end{figure}

    \begin{lemma}
    \label{lem:finding_parallel_calipers}
        For any matching $(f, g)$ between $R$ and $B$, there exists a value $t \in [0, 1]$ and parallel lines tangent to $R(f(t))$ and $B(g(t))$ with $P$ in the area between them.
    \end{lemma}
    \begin{proof}
        Let $\ell_R$ and $\ell_B$ be two coinciding lines tangent to $P$ at $R(1) = B(1)$.
        Due to the continuous and monotonic nature of matchings, we can rotate $\ell_R$ and $\ell_B$ clockwise and counter-clockwise, respectively, around $P$, until they coincide again, such that at every point of the movement, there are points $r \in \ell_R$ and $b \in \ell_B$ that are matched by $(f, g)$.
        See~\cref{fig:rotating_calipers} for an illustration.
        Because the lines start and end as coinciding tangents, there must be a point in time strictly between the start and end of the movement where the lines are parallel.
        The area between these lines contains $P$, and thus these lines specify a time $t \in [0, 1]$ that satisfies the claim.
    \end{proof}

    \begin{lemma}
    \label{lem:maximally-parallel}
        There exists a \f matching between $R$ and $B$ that is a maximally-parallel matching with respect to a line perpendicular to an edge of $P$.
    \end{lemma}
    \begin{proof}
        Let $(f, g)$ be an arbitrary \f matching.
        Consider two parallel tangents $\ell_R$ and $\ell_B$ of $P$ with $P$ between them, such that exists a $t \in [0, 1]$ for which $R(f(t))$ lies on $\ell_R$ and $B(g(t))$ lies on $\ell_B$.
        Such tangents exist by~\cref{lem:finding_parallel_calipers}.
        Let $r^* \in \ell_R \cap R$ and $B(y) \in \ell_B \cap B$ form a bichromatic closest pair of points and let $\ell^*$ be the line through them.
        We show that the maximally-parallel matching $(f^*, g^*)$ with respect to $\ell^*$ is a \f matching
        
        Let $R[r, r']$ and $B[b, b']$ be the maximal subcurves for which the lines supporting $\overline{rb}$ and $\overline{r'b'}$ are parallel to $\ell$, and $R[r, r']$ and $B[b, b']$ are contained in the strip bounded by these lines.
        The parallel part of $(f^*, g^*)$ matches $R[r, r']$ to $B[b, b']$ such that for every pair of points matched, the line through them is parallel to $\ell^*$.
        For every pair of matched points $\hat{r} \in R[r, r']$ and $\hat{b} \in B[b, b']$, we naturally have $d(\hat{r}, \hat{b}) \leq d(r^*, b^*)$.
        By virtue of $r^*$ and $b^*$ forming a bichromatic closest pair among all points on $\ell_R \cap R$ and $\ell_B \cap B$, we additionally have $d(r^*, b^*) \leq d(R(f(t)), B(g(t))) \leq \dF(R, B)$.
        Hence the parallel part has a cost of at most $\dF(R, B)$.

    \begin{figure}
        \centering
        \includegraphics[page=4]{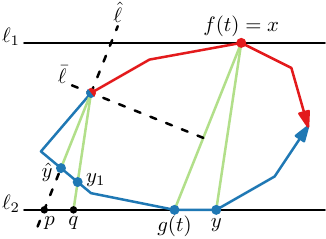}
        \caption{An illustration of the various points and lines used in~\cref{lem:maximally-parallel}.}
        \label{fig:maximally-parallel_construction}
    \end{figure}

        Next we prove that the costs of the first and last fans of $(f^*, g^*)$ are at most $\dF(R, B)$.
        We prove this for the first fan, which matches the prefix of $R$ up to $r$ to the prefix of $B$ up to $b$; the proof for the other fan is symmetric.
        We assume without loss of generality that $r = R(1)$, so $(f^*, g^*)$ matches $r$ to the entire prefix of $B$ up to $b$.
        Let $r_t = R(f(t))$ and $b_t = B(g(t))$.
        See~\cref{fig:maximally-parallel_construction} for an illustration of the various points and lines used.

        Let $\hat{\ell}$ (respectively $\bar{\ell}$) be the line through $r$ that is parallel (respectively perpendicular) to the line through $r_t$ and $b_t$.
        The lines $\hat{\ell}$ and $\bar{\ell}$ divide the plane into four quadrants.
        Let $\hat{B}$ be the maximum prefix of $B$ that is interior-disjoint from $\hat{\ell}$.
        The subcurves $R[1, f(t)]$ and $\hat{B}$ lie in opposite quadrants.
        Hence for each point on $\hat{B}$, its closest point on $R[1, f(t)]$ is $r$.
        Since $b_t$ does not lie interior to $\hat{B}$, the \f matching $(f, g)$ matches all points on $\hat{B}$ to points on $R[1, f(t)]$.
        It follows that the cost of matching $r$ to all of $\hat{B}$ is at most $\dF(R, B)$.

        To finish our proof we show that the cost of matching $r$ to the subcurve $B'$, starting at the last endpoint of $\hat{B}$ and ending at $b$, is at most $\dF(R, B)$.
        We consider two triangles.
        The first triangle, $\Delta$, is the triangle with vertices at $r_t$, $b_t$, and $b^*$.
        For the second triangle, let $p$ be the point on $\ell_B$ for which $\overline{rp}$ lies on $\bar{\ell}$, and let $q$ be the point on $\ell_B$ for which $b \in \overline{rq}$.
        We define the triangle $\hat{\Delta}$ to be the triangle with vertices at $r$, $p$, and $q$.
        The two triangles $\Delta$ and $\hat{\Delta}$ are similar, with $\Delta$ having longer edges.
        Given that $\overline{r_t b_t}$ is the longest edge of $\Delta$, it follows that all points in $\Delta$ are within distance $d(R(f(t)), B(g(t))) \leq \dF(R, B)$ of $r_t$.
        By similarity we obtain that all points in $\hat{\Delta}$ are within distance $\dF(R, B)$ of $r$.
        The subcurve $B'$ lies inside $\hat{\Delta}$, so the cost of matching $r$ to all of $B'$ is at most $\dF(R, B)$.
        This proves that the cost of the first fan, matching $r$ to the prefix of $B$ up to $b$, is at most~$\dF(R, B)$.
    \end{proof}

    Next we give a linear-time algorithm for constructing a maximally-parallel \f matching.
    First, note that there are only $\bigO(n+m)$ maximally-parallel matchings of the form given in~\cref{lem:maximally-parallel} (up to reparameterizations).
    This is due to the fact that there are only $\bigO(n+m)$ pairs of parallel tangents of $P$ whose intersection with $P$ is distinct~\cite{shamos78computational_geometry}.
    With the method of~\cite{shamos78computational_geometry} (nowadays referred to as ``rotating calipers''), we enumerate this set of pairs in $\bigO(n+m)$ time.

    We consider only the pairs of lines $\ell_R$ and $\ell_B$ where $\ell_R$ touches $R$ and $\ell_B$ touches $B$.
    Let $(\ell_{R, 1}, \ell_{B, 1}), \dots, (\ell_{R, k}, \ell_{B, k})$ be the considered pairs.
    For each considered pair of lines $\ell_{R, i}$ and $\ell_{B, i}$, we take a bichromatic closest pair formed by points $r^*_i \in \ell_{R, i} \cap R$ and $b^*_i \in \ell_{B, i}$.
    We assume that the pairs of lines are ordered such that for any $i \leq i'$, $r^*_i$ comes before $r^*_{i'}$ along $R$ and $b^*_i$ comes after $b^*_{i'}$ along $B$.
    We let $(f_i, g_i)$ be the maximally-parallel matching with respect to the line through $r^*_i$ and $b^*_i$.
    By~\cref{lem:maximally-parallel}, one of these matchings is a \f matching.
    To determine which matching is a \f matching, we compute the costs of the three parts (the first fan, the parallel part, and the last fan) of each matching $(f_i, g_i)$.
    
    The cost of the parallel part of $(f_i, g_i)$ is equal to $d(r^*_i, b^*_i)$.
    We compute these costs in $\bigO(n+m)$ time altogether.
    For the costs of the first fans, suppose without loss of generality that there exists an integer $j$ for which the first fan of $(f_i, g_i)$ matches $R(1)$ to a prefix of $B$ for all $i \leq j$, and matches $B(1)$ to a prefix of $R$ for all $i > j$.
    As $i$ increases from $1$ to $j$, the prefix of $B$ matched to $R(1)$ shrinks.
    Through a single scan over $B$, we compute the cost function $y \mapsto \max_{\hat{y} \in [1, y]} d(R(1), B(\hat{y}))$, which measures the cost of matching $R(1)$ to any given prefix of $B$.
    This function is piecewise hyperbolic with a piece for every edge of $B$.
    Constructing the function takes $\bigO(m)$ time, and allows for computing the cost of matching any given prefix to $R(1)$ in constant time.
    We compute the first fan of $(f_i, g_i)$ for all $i \leq j$ in $\bigO(m)$ time by scanning backwards over $B$.
    Extracting the costs of these fans then takes $\bigO(j)$ additional time in total.

    Through a procedure symmetric to the above, we compute the cost of the first fan of $(f_i, g_i)$ for all $i > j$ in $\bigO(n)$ time altogether, through two scans of $R$.
    Thus, the cost of all first fans, and by symmetry the costs of the last fans, can be computed in $\bigO(n+m)$ time.
    Taking the maximum between the costs of the first fan, the parallel part, and the last fan, for each matching $(f_i, g_i)$, we obtain the cost of the entire matching.
    Picking the cheapest matching yields a \f matching between $R$ and $B$.

    \begin{theorem}
        Let $R \from [1, n] \to \R^2$ and $B \from [1, m] \to \R^2$ be two simple curves bounding a convex polygon, with $R(1) = B(1)$ and $R(n) = B(m)$.
        We can construct a \f matching between $R$ and $B$ in $\bigO(n+m)$ time.
    \end{theorem}

\section{The Fr\'echet distance between separated one-dimensional curves}
\label{sub:linear_time_frechet}

    \Cref{lem:advancing_to_prefix-minima,lem:prefix_minimum_matching} show that we can be ``oblivious'' when constructing prefix-minima matchings.
    Informally, for any $\delta \geq \dF(R, B)$, we can construct a prefix-minima $\delta$-matching by always choosing an arbitrary curve to advance to the next prefix-minima, as long as we may do so without increasing the cost past $\delta$.
    We use this fact to construct a \f matching between $R$ and $B$ (which do not have to end in prefix-minima) in $\bigO(n+m)$ time:

    \begin{theorem}
    \label{thm:computing_Frechet_matching}
        Let $R$ and $B$ be two separated one-dimensional curves with $n$ and $m$ vertices.
        We can construct a \f matching between $R$ and $B$ in $\bigO(n+m)$ time.
    \end{theorem}
    \begin{proof}
        Let $R(i^*)$ and $B(j^*)$ form a bichromatic closest pair of points.
        \Cref{cor:matching_closest_pair} shows that there exists a \f matching that matches $R(i^*)$ to $B(j^*)$.
        The composition of a \f matching between $R[1, i^*]$ and $B[1, j^*]$, and a \f matching between $R[i^*, n]$ and $B[j^*, m]$ is therefore a \f matching between $R$ and $B$.

        We identify a bichromatic closest pair of points in $\bigO(n+m)$ time, by traversing each curve independently.
        Next we focus on constructing a \f matching between $R[1, i^*]$ and $B[1, j^*]$.
        The other matching is constructed analogously.

        Let $\delta = \dF(R, B)$.
        \Cref{lem:prefix_minimum_matching} shows that there exists a prefix-minima $\delta$-matching between $R[1, i^*]$ and $B[1, j^*]$.
        If $i^* = 1$ or $j^* = 1$, then this matching is trivially a \f matching.
        We therefore assume $i^* > 1$ and $j^* > 1$.
        Let $R(i)$ and $B(j)$ be the second prefix-minima of $R$ and $B$ (the first being $R(1)$ and $B(1)$), and observe that $i \leq i^*$ and $j \leq j^*$.
        Any prefix-minima matching must match $R(i)$ to $B(1)$ or $R(1)$ to $B(j)$.

        By~\cref{lem:advancing_to_prefix-minima}, there exist $\delta$-matchings between $R[i, n]$ and $B$, as well as between $R$ and $B[j, m]$.
        Thus, if $\dF(R[1, i], B(1)) \leq \delta$, we may match $R[1, i]$ to $B(1)$ and proceed to construct a \f matching between $R[i, n]$ and $B$.
        Symmetrically, if $\dF(R(1), B[1, j]) \leq \delta$, we may match $R(1)$ to $B[1, j]$ and proceed to construct a \f matching between $R$ and $B[j, m]$.
        In case both hold, we may choose either option.

        One issue we have to overcome is the fact that $\delta$ is unknown.
        However, we of course have $\min\{\dF(R[1, i], B(1)), \dF(R(1), B[1, j])\} \leq \delta$.
        Thus the main algorithmic question is how to efficiently compute these values.
        For this, we implicitly compute the costs of advancing a curve to its next prefix-minimum.

        Let $R(1) = R(i_1), \dots, R(i_k) = R(i^*)$ and $B(1) = B(j_1), \dots, B(j_\ell) = B(j^*)$ be the sequences of prefix-minima of $R$ and $B$.
        We explicitly compute the values $\max\limits_{x \in [i_{k'}, i_{k'+1}]} |R(x)|$ and $\max\limits_{y \in [j_{\ell'}, j_{\ell'+1}]} |B(y)|$ for all $1 \leq k' \leq k-1$ and $1 \leq \ell' \leq \ell-1$.
        These values can be computed by a single traversal of the curves, taking $\bigO(n+m)$ time.

        The cost of matching $R[i_{k'}, i_{k'+1}]$ to $B(j_{\ell'})$ is equal to
        \[
            \dF(R[i_{k'}, i_{k'+1}], B(j_{\ell'})) = \max\limits_{x \in [i_{k'}, i_{k'+1}]} |R(x)| + |B(j_{\ell'})|.
        \]
        With the precomputed values, we can compute the above cost in constant time.
        Symmetrically, we can compute the cost of matching $R(i_{k'})$ to $B[j_{\ell'}, j_{\ell'+1}]$ in constant time.
        Thus we can decide which curve to advance in constant time, giving an $\bigO(i^*+j^*)$ time algorithm for constructing a \f matching between $R[1, i^*]$ and $B[1, j^*]$.
    \end{proof}

\end{document}